\documentclass[a4paper,10pt]{article}
\usepackage[utf8]{inputenc}
\usepackage{amsfonts}
\usepackage{epstopdf}
\usepackage{graphicx}
\usepackage{bm}
\usepackage{xcolor}
\usepackage{bbold}
\usepackage{mathtools}
\usepackage{float}
\usepackage{extarrows}
\usepackage{graphicx}
\usepackage{calligra,mathrsfs}
\usepackage{amssymb}
\usepackage{tikz-cd}
\usepackage{enumitem}
\usepackage[normalem]{ulem}
\usepackage{faktor}
\usepackage{amsthm}
\usepackage{listings}
\usepackage{xr}
\usepackage{authblk}
\usepackage{hyperref}

\newtheorem*{theorem*}{Theorem}
\newtheorem*{lem}{Lemma}
\newtheorem*{claim*}{Claim}
\newtheorem*{corr*}{Corollary}
\newtheorem{definition}{Definition}
\numberwithin{definition}{section}
\newtheorem{remark}{Remark}
\numberwithin{remark}{section}
\newtheorem{theorem}{Theorem}
\numberwithin{theorem}{section}
\newtheorem{lemma}{Lemma}
\numberwithin{lemma}{section}
\newtheorem{corollary}{Corollary}
\numberwithin{corollary}{section}
\newtheorem{proposition}{Proposition}
\numberwithin{proposition}{section}
\newtheorem{example}{Example}
\numberwithin{example}{section}

\DeclareMathOperator*{\argmin}{arg\,min}

\newcommand{\diag}{\mathrm{diag}}
\newcommand{\R}{\mathbb{R}}
\newcommand{\dd}{\mathrm{d}}

\newcommand{\T}{T}

\newcommand{\Img}{\mathrm{Im}}
\newcommand{\Ker}{\mathrm{Ker}}

\newcommand{\RqT}{\left( \mathbb{R}^q \right)^*}
\newcommand{\RnT}{\left( \mathbb{R}^n \right)^*}

\newcommand{\xeq}{x_{\mathrm{eq}}}

\newcommand{\V}{\mathcal{V}}
\newcommand{\LT}{\mathfrak{l}}
\newcommand{\bX}{\overline{X}}
\newcommand{\bx}{\overline{x_0}}
\newcommand{\bU}{\overline{U}}
\newcommand{\A}{A}
\newcommand{\Z}{\mathcal{C}}
\newcommand{\Q}{Q}
\newcommand{\UU}{U}
\newcommand{\uu}{u}

\newcommand{\HH}{H}
\newcommand{\X}{\diag(x)}
\newcommand{\XX}{\diag \left( \frac{1}{x} \right)}

\newcommand{\ai}{\alpha_i } 
\newcommand{\gH}{g_H}
\newcommand{\gl}{g_{Q}}

\makeatletter
\newcounter{savesection}
\newcounter{apdxsection}
\renewcommand\appendix{\par
  \setcounter{savesection}{\value{section}}%
  \setcounter{section}{\value{apdxsection}}%
  \setcounter{subsection}{0}%
  \gdef\thesection{\@Alph\c@section}}
\newcommand\unappendix{\par
  \setcounter{apdxsection}{\value{section}}%
  \setcounter{section}{\value{savesection}}%
  \setcounter{subsection}{0}%
  \gdef\thesection{\@arabic\c@section}}
\makeatother

\title{Information geometry of chemical reaction networks: Cramer-Rao bound and absolute sensitivity revisited}
\author[1]{Dimitri Loutchko\thanks{d.loutchko@edu.k.u-tokyo.ac.jp}}
\author[2]{Yuki Sughiyama}
\author[1,3]{Tetsuya J. Kobayashi\thanks{http://research.crmind.net}}
\affil[1]{{Institute of Industrial Science, The University of Tokyo, 4-6-1, Komaba, Meguro-ku, Tokyo 153-8505 Japan.}}
\affil[2]{Graduate School of Information Sciences, Tohoku University, Sendai 980-8579, Japan}
\affil[3]{Department of Mathematical Informatics,
Graduate School of Information Science and Technology,
The University of Tokyo, Tokyo 113-8654, Japan.}
\affil[4]{Universal Biology Institute, The University of Tokyo,
7-3-1, Hongo, Bunkyo-ku, 113-8654, Japan.}
\date{}

\begin{document}

\maketitle

\begin{abstract}
Information geometry is based on classical Legendre duality but allows to incorporate additional structure such as algebraic constraints and Bregman divergence functions.
It is naturally suited, and has been successfully used, to describe the thermodynamics of chemical reaction networks (CRNs) based on the Legendre duality between concentration and potential spaces, where algebraic constraints are enforced by the stoichiometry.
In this article, the Riemannian geometrical aspects of the theory are explored.
It is shown that duality between concentration and potential spaces and the natural parametrizations of equilibrium subspace are isometries, which leads to a multivariate Cramer-Rao bound through the comparison of two Riemannian metric tensors.

In the subsequent part, the theory is applied to the recently introduced concept of absolute sensitivity.
Using the Riemannian geometric tools, it is proven that the absolute sensitivity is a projection operator onto the tangent bundle of the equilibrium manifold.
A linear algebraic characterization and explicit results on first order corrections to the thermodynamics of ideal solutions are provided.
Finally, the theory is applied to the IDHKP-IDH glyoxylate bypass regulation system.

The novelty of the theory is that it is applicable to CRNs with non-ideal thermodynamical behavior, which are prevalent in highly crowded cellular environments due to various interactions between the chemicals.
Indeed, the analyzed example shows remarkable behavior ranging from hypersensitivity to negative-self regulations.
These are effects which usually require strongly nonlinear reaction kinetics.
However, here, they are obtained by tuning thermodynamical interactions providing a complementary, and physically well-founded, viewpoint on such phenomena.
\end{abstract}

\section{Introduction}

The theory of chemical reaction networks (CRNs) plays a central role in the understanding of biomolecular processes and mechanisms such as circadian clocks \cite{gonze2006,hatakeyama2012}, proofreading kinetics \cite{hopfield1974,banerjee2017}, absolute concentration robustness \cite{barkai1997robustness,alon1999robustness,anderson2014stochastic,joshi2023reaction,puente2024absolute,kaihnsa2024absolute,meshkat2022absolute}, perfect adaptation \cite{hirono2023,khammash2021perfect}, hypersensitivity \cite{acar2008stochastic,tyson2008biological,reyes2022numerical}, among many others \cite{mikhailov2017,alon2019}.

The success of CRN theory is rooted in its well-founded and constantly evolving mathematical basis.
Initiated by Horn, Jackson, and Feinberg \cite{horn1972,feinberg1972complex,horn1972necessary,feinberg2019foundations}, the contemporary mathematics of CRN is firmly rooted in the theory of differential equations \cite{feinberg1980chemical,mincheva2008multigraph}, graph theory \cite{craciun2006multiple,craciun2011graph} and algebraic geometry \cite{craciun2009toric,craciun2022disguised,craciun2023structure} but also spreads across various other disciplines such as homological algebra \cite{hirono2021} and, most recently, information geometry.
The information geometrical approach makes use of the analogy between the vector of chemical concentrations and probability distributions on finite spaces and allows to transfer techniques from statistics and information theory to study CRN.
This viewpoint has allowed to establish a geometrical understanding of the thermodynamics of CRN \cite{sughiyama2022hessian,kobayashi2022kinetic,kobayashi2023information} and provided new insights into thermodynamic uncertainty relations and speed limits \cite{loutchko2023geometry}, growing systems \cite{sughiyama2022chemical}, and finite time driving \cite{loutchko2022riemannian}.

In this article, this line of work is extended by exploring the Riemannian geometrical aspects of the information geometry of CRNs.
All spaces are equipped with strictly convex potential functions inducing Legendre duality, and the natural Riemannian metrics are provided by the Hessians of the respective convex functions.
It is then shown that Legendre transformations are isometries with respect to this structure.
In fact, the Riemannian metric tensor can be thought of as encoding the information on the derivative of the Legendre transformation at each point.
This compatibility of derivatives and metric tensors is used throughout the text.
With regard to CRN theory, the relationships between the concentration and potential spaces, the respective equilibrium subspaces and their natural parameter spaces are established and condensed into commutative diagrams.
It is then shown that these diagrams are diagrams of isometries which leads to formulation of a multivariate Cramer-Rao bound for CRNs as the comparison of two Riemannian metric tensors.

In the second part, the concept of {\it absolute sensitivity}, which was introduced in \cite{loutchko2024cramer}, is analyzed through the lens of the Riemannian geometric results.
The classical sensitivity matrix encodes how the steady state concentrations of a CRN change upon infinitesimal perturbations of the linear conserved quantities.
It contains valuable information on the robustness of the CRN but the numerical values of the matrix elements depend on the choice of a basis of $\Ker[S^*]$ (where $S$ is the stoichiometric matrix of the respective CRN) and are therefore difficult to interpret.
In \cite{loutchko2024cramer}, a construction was provided which retains the function of the sensitivity matrix but rectifies the dependence on the choice of a basis.
The essential idea is to trace how an infinitesimal perturbation $\delta x^i$ in the concentration of a chemical $X_i$ of a state $x$ changes the steady state and thus propagates to changes of the concentrations $A^j_i(x) \delta x^i$ of all other chemicals $X_j$ in the adjusted steady state.
The quantity $A^j_i(x)$ is called the absolute sensitivity of $X_j$ with respect to $X_i$.
Such ideas have also been explored from the viewpoint of computational algebra in \cite{feliu2019}.

In \cite{loutchko2024cramer}, a linear algebraic characterization of absolute sensitivity was provided for {\it quasi-thermostatic} CRNs via a linear algebraic Cramer-Rao bound.
This article puts this characterization into a proper information geometric context by promoting the local linear algebra to global tensor algebra and the matrix of absolute sensitivities to a projection operator, respectively.
Most importantly, it becomes possible to study the geometry induced by potential functions of non-ideal systems\footnote{Quasi-thermostatic CRNs are precisely the CRNs which have the geometry of equilibrium CRNs with the thermodynamics of ideal solutions.}.
In doing so, the scope of applicability is greatly enhanced as the effects of crowded environments with, e.g., volume exclusion effects as well as attractive and repulsive interactions between the chemicals can be analyzed.
This is important for biochemical applications, because cellular environments are far from ideal, and yet the CRNs of {\it in vivo} systems are usually described by mass action kinetics and ideal potential functions.
Using the example of the core module of the IDHKP-IDH glyoxylate bypass regulation system, it is demonstrated how kinetically unusual but biologically important effects can be the result of thermodynamical interactions in a crowded environment.

The article is structured as follows:
In the remainder of this introductory section, the mathematical notation is introduced and an outline of the main constructions and results is given.
In Sections \ref{sec:CRN} and \ref{sec:V_construction}, the basic notions of CRN theory and Legendre duality are introduced, and examples of equilibrium manifolds are given.
Section \ref{sec:information_geometry} is the main mathematical section of this article as it contains the analysis of Riemannian geometrical aspects in \ref{sec:Riemannian} and the proof of the Cramer-Rao bound in \ref{sec:Hessian_CRB}.
This is used to derive the main results on absolute sensitivity as a projection operator in Section \ref{sec:abs_sens}, followed by an analysis of first order perturbations of the geometry induced by ideal potential functions in Section \ref{sec:first_order}.
An example is presented in Section \ref{sec:IDH}, with a summary and outlook in Section \ref{sec:discussion}.

\subsection*{Notation}

\paragraph{Vector spaces and duality}
The canonical basis of $\R^n$ is denoted by $\{e_i\}_{i=1}^n$ and the respective dual basis of $(\R^n)^*$ is $\{e^i\}_{i=1}^n$.
A vector $x = \sum_{i=1}^n x^i e_i \in \R^n$ is represented by its coordinates as $x = (x^1,\dotsc,x^n)$, written shortly as $(x^i)$, and $y = \sum_{i=1}^n y_i e^i \in \RnT$ by $y = (y_1,\dotsc,y_n)$ and abbreviated by $(y_i)$.
The bilinear dual pairing $\langle.,.\rangle: \R^n \times \RnT \rightarrow \R$ is thus given by $\langle x,y \rangle = \sum_{i=1}^n x^iy_i$.
\paragraph{Linear maps}
A $n \times m$ matrix $\left[L^i_j\right]_{j=1,\dotsc,m}^{i=1,\dotsc,n}$ is identified with a linear map $L : \R^m \rightarrow \R^n$ as $\epsilon_j \mapsto \sum_{i=1}^n L_j^i e_i$, where $\{e_i\}_{i=1}^n$ and $\{ \epsilon_j \}_{j=1}^m$ are the respective canonical basis vectors of $\R^n$ and $\R^m$.
The adjoint map $L^*: \RnT \rightarrow (\R^m)^*$ to $L$, defined by the property $\langle L y, z \rangle = \langle y , L^* z\rangle$ for $y \in \R^m, z \in (\R^n)^*$, is represented by the transpose matrix $L^*$ via $e^i \mapsto \sum_{j=1}^m L_j^i \epsilon^j$.
\paragraph{Differential geometry}
All manifolds $M$ considered in this text are open submanifolds of some $\R^n$ and therefore the global coordinate systems of $\R^n$ are used.
Let $x = (x^1,\cdots,x^n)$ be a point of $\R^n$.
The tangent space $T_x \R^n$ is spanned by $\left\{ \frac{\partial}{\partial x^i}, \dotsc,  \frac{\partial}{\partial x^n}\right\}$ and its dual, the contangent space $T_x^* \R^n$, is spanned by $\{ \dd x^1, \dotsc,  \dd x^n\}$.
The tangent bundle is given by $T \R^n= \sqcup_{x \in \R^n} T_x \R^n$ with points $(x,v) \in T\R^n$ written as $x = \sum_{i=1}^n x^i e_i$ and $v = \sum_{i=1}^n v^i \frac{\partial}{\partial x^i}$, and analogously for the cotangent bundle $T^*\R^n= \sqcup_{x \in \R^n} T_x^* \R^n$.
For $M \subset \R^n$, the bundle $TM$ is restriction to $M$ of the subbundle of $T\R^n$ which is tangent to $M$, and $T^*M$ is its dual.

A (pseudo-)Riemannian metric on $\R^n$ is given by a symmetric positive (semi) definite tensor $g = \sum_{i,j = 1}^n g_{ij} \dd x^i \otimes \dd x^j$, which is equivalent to the bundle morphism $g: T\R^n \rightarrow T^*\R^n$, $(x,v) \mapsto \left(x, \sum_{i,j=1}^n v^i g_{ij}(x) \dd x^j \right)$.
Locally at a point $x \in \R^n$, this morphism is identified with the linear map $g(x): T_x \R^n \rightarrow T_x^* \R^n, \frac{\partial}{\partial x^i} \mapsto \sum_{j=1}^n g_{ij}(x) \dd x^j$ and is represented by the ($x$-dependent) matrix $\left[g_{ij}(x)\right]_{i,j=1}^n$.
The global object is denoted by $g$ and its local representation in $x$-coordinates by $g(x)$, with the exception of Sections \ref{sec:Hessian_CRB} and \ref{sec:first_order_comp} for notational convenience.

The same applies to the derivative and Jacobian matrices:
For a map $f: M \rightarrow N$ between manifolds, the derivative is the bundle morphism $Df:TM \rightarrow TN$, explicitly given by $(x,v) \mapsto \left(y=f(x), \sum_{i=1}^n \sum_{j=1}^m v^i \frac{\partial y^j}{\partial x^i} \frac{\partial}{\partial y^j} \right)$.
Locally at $x \in M$, the Jacobian is the linear map $D_xf: T_xM \rightarrow T_{y=f(x)}N$ represented by the matrix $\left[D_xf\right]^j_i = \frac{\partial y^j}{\partial x^i}$.
More generally, for any bundle morphism $A$, its local representation at $x$ is denoted by $A(x)$, with the exception of Sections \ref{sec:Hessian_CRB} and \ref{sec:first_order_comp}.

\subsection*{Mathematical Summary}

In this article, the focus is put on equilibrium CRNs, following \cite{sughiyama2022hessian}.
However, because the results are based on geometry, they are valid for any steady state manifold with the same geometry.
This turns out to be valid (at least locally) for a wide class of steady states \cite{agazzi2018geometry,snarski2021hamilton}.
For the thermodynamics of ideal solutions the equilibrium condition is equivalent to the Wegschneider condition in chemical kinetics, cf. \cite{kobayashi2022kinetic}.
For non-ideal solutions, more intricate geometries of the equilibrium states are obtained.

\paragraph{Geometrical setup}
Consider a CRN with $n$ chemicals $X_1,\dotsc,X_n$ and $r$ reactions, let $X$ be a convex subspace of $\R^n_{>0}$ of full dimension, called the {\it concentration space}, and $S$ the stoichiometric matrix of the CRN.
The dynamics of the CRN is given by $\frac{\dd x}{\dd t} = Sj(x)$ for some flux vector $j(x) \in \R^r$ which restricts the trajectories for the initial condition $x_0 \in X$ to the stoichiometric polytope
\begin{equation*}
    P(x_0) := (x_0 + \Img[S]) \cap X.
\end{equation*}
The stoichiometric polytopes give a fibration of the concentration space with the tautological base space $\overline{X}:= X/\Img[S]$,
which comes with a projection map $p: X \rightarrow \overline{X}$.
The elements of $\overline{X}$ are known as {\it linear conserved quantities} in the CRN literature because they are invariants of the CRN dynamics.

Another natural base of the fibration is given by the manifold of equilibrium points, called the {\it equilibrium manifold} which is constructed following \cite{sughiyama2022hessian}:
For an equilibrium CRN, there exists a strictly convex free energy function $\phi \in C^2(X,\R)$ which induces a one-to-one correspondence with the chemical potential space $Y \subset \RnT$ via Legendre duality $X \ni x \mapsto y(x) := \mathrm{argmax}_{y \in Y} [\langle x,y \rangle - \phi(x)] \in Y$.
This map is abbreviated as $\LT_X : X \rightarrow Y$.
The equilibrium points $y \in Y$ are characterized\footnote{
In thermodynamics, a more familiar characterization of equilibrium points is the variational principle of free energy minimization which assigns to each $x_0 \in X$ the equilibrium point $\xeq(x_0) := \argmin_{x \in P(x_0)} \mathrm{D}[x\| y^0]$, where $\mathrm{D}[. || .] : X \times Y \rightarrow \R$ is the Bregman divergence.
In \cite{sughiyama2022hessian}, it is shown that the variational characterization is equivalent to $\xeq (x_0)$ being the unique point in the intersection $P(x_0) \cap \V$.
\label{foot:variational}
} 
by the condition $y \in \Z := \left( y^0 +\Ker[S^*] \right) \cap Y$ for a base point $y^0 \in Y$, and the equilibrium manifold is therefore given by the inverse Legendre transform:
\begin{equation*}
    \V := \LT^{-1}_X (\Z).
\end{equation*}

Choosing a basis $\{u^j\}_{j=1}^q \subset \RnT$ of $\Ker[S^*]$ yields the map $U = (u^1,\dotsc,u^q)^T : X \rightarrow \R^q$ which descends to an injective map $\overline{U}: \overline{X} \rightarrow \R^q$ and which therefore gives the coordinate system $H:= \overline{U}\left(\overline{X}\right)$ for the space of conserved quantities $\overline{X}$. 
Moreover, using the space of coordinates $H$ for $\overline{X}$, the projection $X \rightarrow \overline{X}$ becomes the linear map $U: X \rightarrow H$.
This gives the adjoint map $U^*: \RqT \rightarrow \RnT$, and its shift $f: \RqT \rightarrow \RnT, \lambda \mapsto y^0 + U^*\lambda$ yields the global coordinate system $Q: = f^{-1} (\Z)$ for the space of equilibrium potentials $\Z$.
The following commutative diagram of spaces (Diagram (\ref{eq:diag_comm}) in the main text) summarizes the above discussion:
\begin{equation} \label{eq:diag_comm_intro}
    \begin{tikzcd}[nodes in empty cells]
    & & (X,\phi) \ar[rr,shift left=.55ex,"\LT_{X}"] \ar[ddll,swap,"p"] & & (Y,\phi^*) \ar[ll,shift left=.55ex,"\LT^{-1}_{X}"] \\
    & & (\mathcal{V},\left.\phi\right|_{\mathcal{V}}) \arrow[u, hook] \ar[d, shift right=.55ex, swap, "U"] \ar[rr,shift left=.55ex,"\LT_X"]  & &  
    (\Z,\left.\phi^*\right|_{\Z}) \arrow[u, hook] \ar[ll,shift left=.55ex,"\LT^{-1}_X"]\\
    \overline{X} \ar[rr, swap,"\overline{U}"] \ar[urr,dotted,swap,"\gamma "] & &(H,\psi)  \ar[rr,shift left=.55ex,"\LT_H"] \ar[u, shift right=.55ex, swap, dotted, "\beta"] & & (\Q, \psi^*) \ar[ll,shift left=.55ex,"\LT^{-1}_H"] \ar[u, "f"].
\end{tikzcd}
\end{equation}
\noindent The maps $\gamma: \overline{X} \rightarrow X $ and $\beta: H \rightarrow X$ are sections to the maps $p$ and $U$.
They are central objects in this article and are constructed as follows:
Pick a point $\overline{x_0} \in \overline{X}$, lift it to $x_0 \in X$ and let $\gamma(\overline{x_0})$ be the unique equilibrium point in the polytope $P(x_0)$, i.e.,
\begin{equation*}
    \gamma(\overline{x_0}) := \xeq(x_0) = \V \cap P(x_0).
\end{equation*}
This is independent of the choice of lifting.
Giving the coordinates $H$ to the space $\overline{X}$ makes $\gamma$ into a section $\beta: H \rightarrow X$.

Finally, the six spaces in the middle and right columns of the diagram are equipped with strictly convex functions which induce the Legendre dualities $\LT_X$ (discussed above) and $\LT_H$ (induced by the convex function $\psi^*(\lambda) := \phi^*(y^0 + U^*\lambda),\lambda \in Q$ and its Legendre dual $\psi$).
This is the content of Proposition \ref{proposition:Birch}.
This setup for equilibrium CRN is presented in detail in Sections \ref{sec:CRN} and \ref{sec:V_construction}.
As examples, the ideal solution, quadratic corrections and the solution with van der Waals type interactions are discussed in Examples \ref{ex:quasiTS}-\ref{ex:vdW}.

\paragraph{Riemannian structure and Cramer-Rao bound}

In Section \ref{sec:information_geometry}, the Riemannian metric structure of the spaces from Diagram (\ref{eq:diag_comm_intro}) is introduced via the Hessians of the respective potential functions, i.e, 
\begin{equation*}
    g_X := \sum_{i,j=1}^n \frac{\partial^2 \phi(x)}{\partial x^i \partial x^j} ~\dd x^i \otimes \dd x^j,
\end{equation*}
and analogously for the other spaces.
This metric is known as Ruppeiner or Weinhold metric \cite{weinhold1975metric,ruppeiner1979thermodynamics}. 
It is tightly connected to optimal finite time driving and is the thermodynamical analogue of the Fisher information metric \cite{salamon1983thermodynamic,crooks2007measuring,loutchko2022riemannian}.
This metric contains the same data as the linearization of the Legendre transform in the sense that the following diagram of bundle morphisms commutes:
\begin{equation} \label{eq:diag_compatibility_intro}
    \begin{tikzcd}[nodes in empty cells]
    TX \ar[rr,"D\LT"] \ar[dr,swap,"g_X"] & & TY \\
    & T^*X. \ar[ur,swap,"\dd x^i \mapsto \frac{\partial}{\partial y_i}"] & 
\end{tikzcd}
\end{equation}
This compatibility leads to Theorem \ref{theorem:Birch} which states:
\begin{theorem*}
    The Diagram (\ref{eq:diag_comm_intro}) is a commutative diagram of isometries.
\end{theorem*}
Then, by applying the Theorem, after some diagram chasing one obtains an expression for the $g_X$-orthogonal restriction of the metic $g_X$ to $\V$ through the Hessian metric $g_H$ on $H$ in Lemma \ref{lemma:Hessian_metric}:
\begin{lem}
The matrix $U^* \gH(\eta) U$ defines a pseudo-Riemannian metric on $X$.
The metric agrees with $g_X$ on $T\mathcal{V}$ and vanishes on its $g_X$-orthogonal complement.
\end{lem}
Note that the matrix $U^* \gH(\eta) U$ represents the pullback of the metric $g_H$ to $\V$ via the map $U$.
This leads to the final result of the section on Riemannian geometry in Theorem \ref{thm:Hessian_CRB}, which states the multivariate Cramer-Rao bound for CRN:
\begin{theorem*}(Hessian geometric Cramer-Rao bound)
For each $x \in X$, the matrix inequality
\begin{equation*}
    g_X(x) \geq U^* \gH(\eta) U
\end{equation*}
holds, i.e., the difference between the two matrices is positive semidefinite.
The matrix $\gH(\eta)$ is explicitly given by $\gH(\eta) = \left[\UU g_X(x)^{-1} \UU^*\right]^{-1}$.
\end{theorem*}
The rest of the article establishes the link between the Riemannian geometry and absolute sensitivity.

\paragraph{Absolute sensitivity}

Sensitivity aims at evaluating how the equilibrium concentrations will change upon perturbations of the linear conserved quantities.
It is thus captured by the derivative $D\gamma: T\overline{X} \rightarrow TX$.
However, to get numerical values for the sensitivity at a point $\overline{x_0} \in \overline{X}$, one has to fix a coordinate system $H$ of $\overline{X}$, with coordinates $\eta$ for $\overline{x_0}$, and compute the Jacobian matrix $D_{\eta}\beta = \frac{\partial x}{\partial \eta}$.
But there is no canonical choice for the coordinate system $H$, and yet the physical interpretation of sensitivity depends on the numerical values of $D_{\eta}\beta$ in applications.
This is remedied by absolute sensitivity which is canonically defined and thus gives physical meaning to the respective numerical values.
The modification is that, instead of perturbing $\eta \in H$, one perturbs the initial condition $x_0$ and traces how it reflects in the change of the equilibrium concentrations of $\xeq(x_0)$.
In other words, absolute sensitivity $A$ is given by the composition of derivatives:
\begin{equation*}
    A:= D(\gamma \circ p) : TX \longrightarrow \overline{X} \longrightarrow TX.
\end{equation*}
Now, on $X$, there is a canonical choice of a coordinate system and thus, at each point $x \in X$, the Jacobian matrix $A(x) = D_x(\gamma \circ p)$ is well-defined.
This is the matrix of absolute sensitivities, and it is introduced in Section \ref{sec:abs_sens}.
Because $A(x)$ is identical for all $x$ within the stoichiometric polytope $P(x)$, it is enough to consider the restriction of $A$ to $\V$ inside $X$, equipped with the pullback bundle of $TX$ to $\V$ via the inclusion $\V \hookrightarrow X$.
Moreover, over points $x \in \V$, the morphism $A$ has a natural interpretation as a projection operator which is now discussed.

\paragraph{Application of Riemannian geometry to absolute sensitivity}

By chasing derivatives in Diagram (\ref{eq:diag_comm_intro}), together with the compatibility (\ref{eq:diag_compatibility_intro}) between the derivative of a Legendre transform and the respective Hessian metric, an explicit expression for the absolute sensitivity is derived as
\begin{equation*}
    A = g_Y U^* g_H U
\end{equation*}
One checks that this is an idempotent operator and thus corresponds to a projection.
Indeed, this is the content of Theorem \ref{thm:abs_sens_final} which states:
\begin{theorem*} \label{thm:abs_sens_final_intro}
    At a point $x \in \V$, the matrix of absolute sensitivities $A(x)$ represents the $g_X$-orthogonal projection $\pi: T_xX \rightarrow T_x \V$ and its matrix elements are
    \begin{equation*}
        A(x)^i_j = \left\langle \dd x^i , \pi\left( \frac{\partial}{\partial x^j} \right) \right\rangle.
    \end{equation*}    
\end{theorem*}
The proof relies on the Cramer-Rao bound.
It is worth noting that, by definition of $A$, it is clear that it will map to the tangent bundle $T\V$ inside $TX$, and that it will be a projection as it restricts to the identity on $T\V$.
The difficulty lies in establishing the $g_X$-orthogonality of the projection which requires the geometrical setup from the previous sections.

For the potential function of the ideal solution, the results from \cite{loutchko2024cramer} on quasi-thermostatic CRNs are recovered, giving the metric $g^{\mathrm{id}}_X(x) = \XX$, and
\begin{equation*}
   \left[ A^{\mathrm{id}}(x) \right]^i_j = \left[ \pi\left(\frac{\partial}{\partial x^j}\right) \right]^i
\end{equation*}
as the resulting matrix of absolute sensitivities.
This yields the restriction $\left[ A^{\mathrm{id}}(x) \right]^i_i \in [0,1]$ on the diagonal terms as a corollary and implies that ideal equilibrium systems cannot exhibit hypersensitivity or negative feedback.

Subsequently, the effect of non-ideal behavior is taken into account by a correction to the metric $g^{\mathrm{id}}_X$ as
\begin{equation*}
    g_X = g_X^{\mathrm{id}}+ C
\end{equation*}
Then, up to first order in the components $C_{ij}$, the expression
\begin{equation*}
    A = \left[ I - (A^{\perp})^{\mathrm{id}} g_Y^{\mathrm{id}} C \right] A^{\mathrm{id}}
\end{equation*}
for the absolute sensitivity is derived.
Lemma \ref{lemma:A_vanishing_corr} gives a condition on the vanishing of the correction term which implies Corollary \ref{corr:first_order}:
\begin{corr*}
    If the total concentration $\sum_{i=1}^n x^i$ is conserved, i.e. if $$ \sum_{i=1}^n e^i \in \Ker[S^*],$$ and if $C_{ij}(x) = c(x)$ for some function $c(x)$, then the absolute sensitivity $A$ coincides with the absolute sensitivity for the corresponding quasi-thermostatic CRN, i.e.,
    \begin{equation*}
        A = A^{\mathrm{id}},
    \end{equation*}
    to first order in $c(x)$.
\end{corr*}

Finally, in Section \ref{sec:IDH}, the results are applied to the core module of the IDHKP-IDH (IDH = isocitrate dehydrogenase, KP = kinase-phosphatase) glyoxylate bypass regulation system as an example.

\section{Chemical reaction networks} \label{sec:CRN}

In this section, chemical reaction networks (CRNs) are introduced, with a particular focus on linear conserved quantities and the classical sensitivity matrix.

\subsection{Basic notions}

The datum of a chemical reaction network (CRN) is given by a set of $n$ chemicals $X_1,\dotsc,X_n$ and $r$ reactions $R_1,\dotsc,R_r$ as
\begin{equation*}
    R_j: \sum_{i=1}^n (S^-)^{i}_{j} X_i \rightarrow \sum_{i=1}^n (S^+)^i_{j} X_i,
\end{equation*}
whereby $(S^-)^{i}_{j}$ and $(S^+)^{i}_{j}$ are the nonnegative integer stoichiometric coefficients.
This is encoded in the $n \times r$ stoichiometric matrix $S$ as
\begin{align*}
    S^i_{j} = (S^+)^{i}_{j}  - (S^-)^{i}_{j}.
\end{align*}
The state space of the CRN is $X = \R^n_{>0}$, whereby a point $x = (x^1,\dotsc,x^n)^T \in X$ represents the concentrations $x^i$ of the chemicals $X_i$.

The dynamics is governed by the differential equation
\begin{equation} \label{eq:dynamics}
    \frac{\dd x}{\dd t} = Sj(x),
\end{equation}
where $j(x) \in \R^r$ is a flux vector.
Usually, it is given by mass action kinetics for elementary reactions \cite{craciun2009toric,feinberg2019foundations} while Michaelis-Menten and Hill type kinetics are viable choices to take into account enzymatically catalyzed reactions and cooperative effects \cite{marin2019kinetics}.
However, in this article, the geometry is derived from thermodynamical arguments and it is not necessary to specify the functional form of $j(x)$.
For any initial condition $x_0 \in X$, the stoichiometric polytope
\begin{equation} \label{eq:def_Px_0}
    P(x_0) := \{ x \in X | x - x_0 = \Img[S] \}
\end{equation}
is invariant under the dynamics (\ref{eq:dynamics}).
The vector space $\R^n$ has a projection to $\R^n / \Img[S]$.
Let $\overline{X}$ denote the image of $X$ under this projection, and let
\begin{equation} \label{eq:pX}
   p : X \longrightarrow \overline{X} :=  \faktor{X}{\Img[S]}
\end{equation}
be the projection map.
The space $\overline{X}$ called the space of {\it linear conserved quantities}.
The fibers of the projection $p: X \rightarrow \overline{X}$ are precisely the stoichiometric polytopes, i.e., for any $\overline{x_0} \in \overline{X}$ one has $p^{-1} (\overline{x_0}) = P(x_0)$, and thus $\bX$ is a base space for the fibration of $X$ by stoichiometric polytopes.
An alternative base of the fibration is provided by the equilibrium manifold in Section \ref{sec:V_construction}.

A coordinate system on $\bX$ is constructed based on the orthogonality between $\Img[S] \subset \R^n$ and $\Ker[S^*] \subset \RnT$:
Let $q$ be the dimension of $\Ker[S^*]$ and let $\{\uu^j\}_{j=1}^q$ be a basis of $\Ker[S^*]$.
This determines a map
\begin{equation} \label{eq:eta_def}
    U := (u^1,\dotsc,u^q)^T : \R^n \rightarrow \R^q.
\end{equation}
The coefficients of the $u^j$ with respect to the canonical basis of $\RnT$ are $u^j_i$, i.e., $u^j = \sum_{i=1}^n u^j_i e^i$.
Then $U$ is represented by the $q \times n$ matrix $(u^j_i)$ and the map (\ref{eq:eta_def}) is given explicitly by
\begin{equation} \label{eq:Ux_explicit}
    \sum_{i=1}^n x^i e_i \mapsto \sum_{i=1}^n \sum_{j=1}^q x^i u^j_i \epsilon_j,
\end{equation}
where $\{\epsilon_j\}_{j =1}^q$ is the canonical basis of $\R^q$.
The kernel of $U$ is $\Img[S]$ and therefore $U$ factors through $\R^n / \Img[S]$, i.e., the diagram
\begin{equation} \label{eq:U_factor}
    \begin{tikzcd}[nodes in empty cells]
    \R^n \ar[rr,"U"] \ar[dr,"p"] & & \R^q \\
    & \faktor{\R^n}{\Img[S]} \ar[ur,"\overline{U}"] & 
\end{tikzcd}
\end{equation}
commutes, and the linear map $\bU$ is a bijection.
The image $U(X)$ is an open submanifold of $\R^q$ of full dimension, denoted by
\begin{align} \label{eq:def_H}
\HH := \UU (X) \subset \R^q.
\end{align}
Using the commutativity of the diagram (\ref{eq:U_factor}) gives the diffeomorphism $\bU: \bX \rightarrow H$ and therefore the inverse $\bU^{-1}: H \rightarrow \bX$ is a global coordinate system for $\bX$.
A vector $\eta \in H$ is called a vector of {\it linear conserved quantities}.
For any $\eta \in \HH$, The stoichiometric polytope is
\begin{align} \label{eq:def_Peta}
    P(\eta) := \{ x \in X| \UU x = \eta \} \subset X.
\end{align}
It agrees with the coordinate free definition (\ref{eq:def_Px_0}):
For any $\overline{x_0} \in \bX$, the polytope $P(x_0)$ is identical to the polytope $P(\eta)$ with $\eta = \bU \overline{x_0}$.

\subsection{Sensitivity} \label{sec:sens_intro}

A general setup to treat sensitivity is to consider differentiable sections
\begin{equation*}
    \gamma : \bX \rightarrow X
\end{equation*}
to the projection $p : X \rightarrow \bX$.
This is due to the fact that all equilibrium manifolds and quasi-thermostatic steady states arise as the image of such a section \cite{sughiyama2022hessian}.

Fix a section $\gamma$ and denote
\begin{equation*}
    \V := \Img[\gamma] \subset X.
\end{equation*}
Equivalently, given such $\V \subset X$, the section $\gamma$ can be constructed as follows:
For any point $\bx \in \bX$, pick a lift $x_0$ to $X$.
Then the section is given by the intersection of $\V$ with the stoichiometric polytope of $x_0$, and this is independent of the choice of the lift:
\begin{equation} \label{eq:sec_gamma}
    \gamma(\bx) = \V \cap P(x_0).
\end{equation}
The {\it sensitivity} of the section $\gamma$ is given by the derivative
\begin{equation*}
    D\gamma : T\bX \longrightarrow TX.
\end{equation*}
Working with the coordinate system $H$ on $\bX$ one yields section
\begin{equation*}
    \beta: H \longrightarrow X
\end{equation*}
to the map $U: X \rightarrow H$, which is given by $\beta = \gamma \circ \bU^{-1}$.
With these coordinates, the {\it sensitivity matrix} at a point $x \in \V$ with $\eta = Ux$ is given by the Jacobian
\begin{equation*}
    D_{\eta} \beta = \frac{\partial x}{\partial \eta}.
\end{equation*}

The sensitivity matrix plays an important role in applications as it is used to evaluate how steady state concentrations change upon perturbations of the linear conserved quantities.
However, it is dependent on the choice of the basis $\{u^j\}_{j=1}^q$ and therefore the respective numerical values of the elements $\frac{\partial x^i}{\partial \eta^j}$ have no canonical significance.
For example, absolute concentration robustness in $X_i$ is equivalent to $\frac{\partial x^i}{\partial \eta^j} = 0$ for all $j = 1,\dotsc,q$ and this is independent of the the choice of basis.
But the condition $\frac{\partial x^i}{\partial \eta^j} \approx 0$ cannot be formulated to characterize approximate concentration robustness in $X_i$ because the $\eta^j$ can be scaled arbitrarily.
This drawback will be fixed in Section \ref{sec:abs_sens} by the introduction of {\it absolute sensitivity}.

In the following section, the focus is put on the construction of the section $\beta$ based on thermodynamical arguments.

\section{The geometry of non-ideal equilibrium manifolds} \label{sec:V_construction}

Often, CRN theory is concerned with studying the manifold $\V \subset X$ of steady states, which is defined by the kinetic condition as the vanishing locus of the vector field $Sj(x)$, cf. Equation (\ref{eq:dynamics}).
However, for equilibrium CRNs, the geometry of $\V$ can be derived from thermodynamical rather than kinetical considerations\footnote{
It should be kept in mind that analogous (local) descriptions of nonequilibrium steady states do exist \cite{anderson2015lyapunov}.
The details are not further discussed in this article as they are rather involved and not essential for the following geometrical considerations.
}.

Thermodynamics is be introduced through a strictly convex potential function $\phi \in C^2(X,\R)$ which induces Legendre duality between $X$ and the space of chemical potentials $Y \subset \RnT$.
In \cite{sughiyama2022hessian}, it has been shown that the equilibrium manifold $\V \subset X$ is the Legendre transform of an affine subspace $\Z \subset Y$.
Moreover, the $\V$ defined this way intersects each stoichiometric polytope $P(x_0)$ in a unique point, which is called the equilibrium point of the polytope\footnote{
It is a basic construction in information geometry known as dual foliation which ensures the existence and uniqueness of the intersection, cf. \cite{amari2016}.
The physically inclined reader might expect a variational characterization of equilibrium states as free energy minimizers instead of the geometrical definition given here.
Indeed, the equilibrium point $\gamma(\bx)$ can be constructed as $\gamma(\bx) = \argmin_{x \in P(x_0)} \mathrm{D}[x \| y^0]$, where $\mathrm{D}[.\|.] : X \times Y \rightarrow \R$ is the Bregman divergence defined as $\mathrm{D}[x \| y] := \phi(x) + \phi^*(y) - \langle x,y\rangle$, the function $\phi^* \in C^2(Y,\R)$ is the Legendre dual of $\phi$, and $y^0$ is any point in $\Z$.
It is again a standard fact in information geometry that the geometrical characterization is equivalent to the variational one.
See \cite{amari2016} for the mathematical background, and \cite{sughiyama2022hessian} for the connection to thermodynamics.
}.
This yields the section $\gamma: \bX \rightarrow X$ which maps each linear conserved quantity to the equilibrium point of the respective polytope, i.e., $\gamma(\bx) := \V \cap P(x_0)$, cf. Equation (\ref{eq:sec_gamma}).

If the CRN obeys the thermodynamics of an ideal solution, then $\phi$ is given by $\phi = \sum_{i =1}^n [\mu_i^0 x^i + x^i \log x^i - x^i]$, where the $\mu_i^0$ are the standard chemical potentials.
In this case, $\V$ is identical to the steady state manifold of a detailed balanced CRN with mass action kinetics\footnote{
Moreover, the steady state manifold of any quasi-thermostatic CRN, and {\it a fortiori} of any complex balanced CRN, has the same geometry as such $\V$.
}.
This $\V$ is an analytified toric variety and the existence and uniqueness of the intersection of $\V$ with any stoichiometric polytope $P(x_0)$ is known as Birch's theorem in algebraic statistics \cite{craciun2008identifiability,craciun2009toric,gross2020joining,craciun2022disguised}.
This is discussed in Example \ref{ex:quasiTS}.
The key point of the information geometrical setup is that this remains valid for equilibrium CRNs even if they are not ideal solutions, as well as for nonequilibrium CRNs whose steady state manifolds have the same geometry \cite{sughiyama2022hessian}.
Therefore, the information geometric setup allows to treat more general shapes than analytified toric varieties while retaining very similar mathematical properties.
One important class are CRNs with van der Waals interactions, presented in Example \ref{ex:vdW}.

\subsection{Legendre duality and equilibrium manifolds} \label{sec:Hessian_setup}

Let $X \subset \R^n_{>0}$ and $Y \subset \RnT$ be convex open subsets of full dimension, and let there be a pair of strictly convex functions $\phi \in C^2(X,\R) $ and $\phi^* \in C^2(Y,\R)$ which induce Legendre duality\footnote{The two convex functions are related by convex conjugation as $\phi(x) = \max_{y \in Y} \{ \langle x,y \rangle - \phi^*(y) \}$ and $\phi^*(y) = \max_{x \in X} \{ \langle x,y \rangle - \phi(x) \}$, and the Legendre transformation between the spaces $X$ and $Y$ is given by the variational formulae $x(y) = \mathrm{argmax}_{x \in X} [\langle x,y \rangle - \phi(x) ]$ and $y(x) = \mathrm{argmax}_{y \in Y} [\langle x,y \rangle - \phi^*(y)]$.} between $X$ and $Y$ (note that the global coordinate systems of $\R^n$ and $\RnT$ are used throughout the text).
The Legendre transformation $\LT_X: X \rightarrow Y$ is given by 
\begin{equation} \label{eq:LT_x}
    x \mapsto y(x) = \sum_{i=1}^n \frac{\partial \phi(x)}{\partial x^i } e^i    
\end{equation}
and its inverse is
\begin{equation} \label{eq:LT_y}
    y \mapsto x(y) = \sum_{i=1}^n \frac{\partial \phi^*(y)}{\partial y_i } e_i.
\end{equation}
The shorthand notation for these maps is
\begin{align}
    \nonumber 
    y &= \LT_X(x)\\
    \nonumber 
    x &= \LT^{-1}_X(y).
\end{align}
To construct the required geometrical setup, consider the linear space in $\RnT$ given by the image of the affine map $f:\RqT \rightarrow \RnT, \lambda \mapsto y^0+ \UU^* \lambda$, where $y^0\in Y$ is an arbitrary base point, $q = \dim \Ker[S^*]$, and $\UU^*$ is the adjoint of the map $U: X \rightarrow \R^q$ which was defined in (\ref{eq:eta_def}).
Define the convex set $\Z:=\Img[f] \cap Y$ and let $\Q:= f^{-1}(\Z) \subset \RqT$ be the corresponding coordinate space for $\Z$: 
\begin{align}
\label{eq:f_def}
    f: \Q &\longrightarrow \Z \subset \RnT \\
\nonumber
    \lambda &\mapsto y^0 + \UU^* (\lambda).
\end{align}
Denote the image of $\Z$ in $X$ via $\LT^{-1}_X$ by $\V$:
\begin{equation} \label{eq:def_V}
    \Q \xrightarrow{f} \Z \xrightarrow{\LT^{-1}_X} \V \subset X.
\end{equation}
This manifold is parametrized by $\Q$ via the composition map $\LT^{-1}_X \circ f$ and this can be used to construct a section $\beta: H \rightarrow \V$ as follows:
The convex function $\phi^*$ remains convex when restricted to $\Z$ and pulls back to a convex function $\psi^*$ on $\Q$, explicitly given by
\begin{equation*}
    \psi^*(\lambda) := f^*\phi^*(\lambda) = \phi^*(y^0+ \UU^* \lambda).
\end{equation*}
The induced Legendre transform $\LT^{-1}_H: \Q \rightarrow \R^q$ is calculated as\footnote{The notation $\LT^{-1}_H$ instead of $\LT_Q$ is used for compatibility with the Legendre duality between $X$ and $Y$ in Diagram (\ref{eq:diag_comm}).}
\begin{align}
\nonumber
    \eta:= \LT^{-1}_H(\lambda) &= \sum_{j=1}^q \frac{\partial \psi^*(\lambda)}{\partial \lambda_j} \epsilon_j =\sum_{j=1}^q \sum_{i=1}^n \frac{\partial \phi^*(y)}{\partial y_i} \frac{\partial[y^0+ \UU^* \lambda]_i}{\partial \lambda_j} \epsilon_j \\
\label{eq:eta_calc}
    &= \sum_{j=1}^q \sum_{i=1}^n x^i u^j_i \epsilon_j = Ux
\end{align}
where $x = \LT^{-1}_X \circ f(\lambda)$ and the formula (\ref{eq:Ux_explicit}) was used.
This leads to the following proposition.
\begin{proposition} \label{prop:H}
    The space of Legendre dual parameters to $\Q$ is $H = \UU(X)$, which is the space of vectors of conserved quantities, as defined in (\ref{eq:def_H}).
\end{proposition}
\begin{proof}
    By the above calculation, the inclusion $U(\V) \subseteq U(X)$ holds.
    It remains to show that, for each $x \in X$, there is a $x' \in \V$ such that $Ux = Ux'$.
    But this follows from the general properties of mixed parametrizations in information geometry, cf. \cite{amari2000methods}, Section 3.7., according to which each point $x \in X$ can be written as $x = x' + x''$, whereby $x' \in \V$ and $x'' \in \Ker[\UU]$.
\end{proof}
The space $H$ is a convex subset of $\R^q$ of full dimension and is equipped with the Legendre dual convex function of $\psi^*$, explicitly given by
\begin{equation*}
    \psi(\eta):= \max_{\lambda \in \Q} \left[ \langle \eta, \lambda \rangle -  \phi^*(y^0 + U^*\lambda) \right].
\end{equation*}
This gives the following commutative diagram of spaces equipped with the respective convex functions:
\begin{equation} \label{eq:diag_comm}
    \begin{tikzcd}[nodes in empty cells]
    (X,\phi) \ar[rr,shift left=.55ex,"\LT_{X}"] & & (Y,\phi^*) \ar[ll,shift left=.55ex,"\LT^{-1}_{X}"] \\
    (\mathcal{V},\left.\phi\right|_{\mathcal{V}}) \arrow[u, hook] \ar[d, shift right=.55ex, swap, "U"] \ar[rr,shift left=.55ex,"\LT_X"]  & &  
    (\Z,\left.\phi^*\right|_{\Z}) \arrow[u, hook] \ar[ll,shift left=.55ex,"\LT^{-1}_X"]\\
    (H,\psi)  \ar[rr,shift left=.55ex,"\LT_H"] \ar[u, shift right=.55ex, swap, dotted, "\beta"] & & (\Q, \psi^*) \ar[ll,shift left=.55ex,"\LT^{-1}_H"] \ar[u, "f"].
\end{tikzcd}
\end{equation}
The main result of this construction is the existence of the section $\beta$.
It is a generalization of Birch's theorem for quasi-thermostatic CRN, cf. \cite{horn1972,craciun2009toric}, to CRN whose steady-state manifold is given by $\V$.
\begin{proposition} \label{proposition:Birch}
The restriction of the map $U$ to $\V$ is a bijection and the inverse map $\beta: H \rightarrow \mathcal{V}$ is uniquely determined by the commutativity of the above diagram.
\end{proposition}
\begin{proof}
    The map $f = y^0+ U^*$ is invertible: It is onto by construction and injective because the the linear map $U^*$ has a trivial kernel.
    From the calculation (\ref{eq:eta_calc}) it follows that the restriction of $U$ to $\V$ is given by $U = \LT^{-1}_H \circ f^{-1} \circ \LT_X$, which is a bijection.
    This bijection is between $\V$ and $H$ according to Proposition \ref{prop:H}.
    Therefore, the section $\beta$ must be given by $\beta = \LT^{-1}_X \circ f \circ \LT_H$ and this determines it uniquely.
\end{proof}
The manifold $\V$ is of central importance to physics because all equilibrium manifolds of CRN have this shape, even for non-ideal thermodynamical behavior \cite{sughiyama2022hessian}.
However, the arguments in the following sections of this article do not rely on this connection to thermodynamics but only on the geometry of $\V$.
Hence, they remain valid for nonequilibrium CRN whose steady state manifolds have the same geometrical shape.
This motivates the following definition:
\begin{definition}
    A {\it quasi-equilibrium manifold} is a manifold $\V \subset X$ which is obtained as the Legendre transform of the affine space $\Z \subset Y$ defined in (\ref{eq:def_V}).
\end{definition}

\subsection{CRNs with ideal behavior, quadratic corrections, and van der Waals interactions} \label{sec:V_examples}

The following three examples illustrate the construction of $\V$.
Example \ref{ex:quasiTS} recovers the well-known geometry of quasi-thermostatic steady state manifold via the potential function $\phi^{\mathrm{id}}(x)$ for ideal solutions.
Example \ref{ex:quadratic} discusses quadratic corrections to $\phi^{\mathrm{id}}(x)$ as the simplest polynomial corrections with a nontrivial influence on the geometry of the CRN.
Example \ref{ex:vdW} provides the construction of $\V$ when the the chemicals exhibit van der Waals type interactions.
These interactions play an important role in crowded environments and therefore are of biological significance. 
\begin{example}[{\it Quasi-thermostatic CRN}] \label{ex:quasiTS}
    The convex potential function of the ideal solution is given by the sum
    \begin{equation} \label{eq:phi_ideal}
        \phi^{\mathrm{id}}(x) = \sum_{i =1}^n [\mu_i^0 x^i + x^i \log x^i - x^i].
    \end{equation}
    It induces the Legendre duality between the concentration space $X = \R^n_{>0}$ and the potential space $Y = \RnT$, given by $x(y) = \exp(y-\mu^0)$ and $y(x) = \mu^0 + \log x$ whereby the exponential and the logarithm are taken componentwise.
    This leads to the ideal equilibrium manifold $\V^{\mathrm{id}} := x_0 \circ \exp\left(U^*\lambda\right)$, where the symbol $\circ$ denotes the Hadamard product, i.e., $(x^i) \in \V^{\mathrm{id}}$ is given by $x^i = x^i_0 \exp\left( \sum_{j=1}^q u^j_i \lambda_j \right)$ with $x^i_0 = \exp (y_i^0-\mu^0_i)$.
    Equivalently, using $\Img[U^*] = \Ker[S^*]$, the manifold $\V^{\mathrm{id}}$ is characterized as
    \begin{equation*}
        \V^{\mathrm{id}} = \{ x \in X | \log x - \log x_0 \in \Ker[S^*] \}.
    \end{equation*}
    CRNs whose steady state manifolds have this form are known as {\it quasi-thermostatic CRNs}.
    The manifolds $\V^{\mathrm{id}}$ appear as steady state manifolds of complex balanced CRN, which also include equilibrium CRN for detailed balanced systems \cite{craciun2009toric,kobayashi2022kinetic}.
\end{example}

\begin{example}[{\it Quadratic correction}] \label{ex:quadratic}
    As a first approximation to more realistic potential functions, the potential function of the ideal solution can be modified by quadratic interaction terms as
    \begin{equation} \label{eq:phi_quad}
        \phi(x) = \phi^{\mathrm{id}}(x) + \sum_{i,j=1}^n a_{ij} x^i x^j.
    \end{equation}
    In physics, the terms $a_{ij} x^i x^j$ correspond to interactions between the chemicals $X_i$ and $X_j$.
    They can be either attractive or repulsive, which determines the sign of the $a_{ij}$.
    This physical interpretation enforces the symmetry condition $a_{ij} = a_{ji}$ for all $i,j$.
    In order to retain the bijectivity of the Legendre transformation, the concentration space $X$ is the restriction of $\R^n_{>0}$ to the submanifold where the Hessian of $\phi(x)$ is positive definite.
    The potential space is the image of $X$ under the map
    \begin{equation} \label{eq:y_quadratic}
        y_i = \frac{\partial \phi(x)}{\partial x^i} = y_i^{\mathrm{id}} + \sum_{j=1}^n [a_{ij} + a_{ji}] x^i.
    \end{equation}
       
\end{example}

\begin{example}[{\it Van der Waals interactions}]  \label{ex:vdW}
    In crowded environments, the CRN shows various deviations from ideal kinetic (mass action kinetics) and thermodynamical behavior.
    From the thermodynamical point of view, such deviations can be approximated by the van der Waals equation of state, which gives rise to the potential function
    \begin{equation} \label{eq:phi_vdW}
        \phi(x) = \phi^{\mathrm{id}}(x) - \sum_{i=1}^n x^i \log \left[ 1 - \sum_{j=1}^n b^i_{j} x^j \right] - \sum_{i,j =1}^n a_{ij}x^i x^j
    \end{equation}
    with $a_{ij}, b^i_{j} \geq 0$, as well as the symmetry $b^i_j = b^j_i$ \cite{vovchenko2017multicomponent}.
    The terms $b^i_{j} x^j$ capture the excluded volume due to the presence of the molecules $X_j$, as felt by the molecules $X_i$, whereas the terms $a_{ij}x^i x^j$ represent interactions between molecules of type $X_i$ and $X_j$ as in the previous example.

    The space $X$ cannot be all of $\R^n_{>0}$ but has to be confined by the conditions $\sum_{j=1}^n b^i_{j} x^j < 1$ as well as by the strict convexity of $\phi(x)$.  
    In the Supplementary Material, it is shown that there is always an open convex submanifold $X$ of $\R^n_{>0}$ satisfying these conditions.
    This manifold is given by
    \begin{equation*}
        X = \left\{ x \in \R^n_{>0}: x^i < \frac{1}{nC_i} \right\},
    \end{equation*}  
    where $C_i:= \max_{j} \left\{ a_{ij} + a_{ji} + 2B, 4B \right\}$ and $B := \max_{i,j} \left\{ b^i_j \right\}$.
    The Legendre dual coordinates are computed as
    \begin{equation} \label{eq:y_vdW}
        y_k = \frac{\partial \phi(x)}{\partial x^k} = y_k^{\mathrm{id}} + \sum_{i=1}^n x^i \frac{b^i_k}{1 - \sum_{j=1}^n b^i_j x^j} - \log \left[1 - \sum_{j=1}^n b^k_j x^j \right] - \sum_{i=1}^n [a_{ik} + a_{ki}]x^i,
    \end{equation}
    with $y_k^{\mathrm{id}} = \mu_k^0 + \log x^k$ as in Example \ref{ex:quasiTS}.
    As in the previous example a general analytic expression for $\V$ does not exist in general as it requires inverting the system of equations (\ref{eq:y_vdW}).
    However, for the computations of absolute sensitivity in Section \ref{sec:abs_sens} only the tangent space $T_x \V$ is of interest which can be computed explicitly as $T_x \V = \frac{\partial x}{\partial y}\Img[U^*] = \left[\frac{\partial^2 \phi(x)}{\partial x \partial x}\right]^{-1} \Img[U^*]$.
\end{example}

In the next section, the Riemannian geometry of the spaces from Diagram (\ref{eq:diag_comm}) is developed.

\section{Riemannian structure and Cramer-Rao bound} \label{sec:information_geometry}

Viewing the vector of chemical concentrations as a distribution analogous to a probability distribution makes information geometrical techniques available in CRN theory and has led to a geometrical understanding of the link between information theory and thermodynamics of CRN \cite{sughiyama2022hessian,kobayashi2022kinetic,kobayashi2023information}.
The information geometry, also termed Hessian geometry in the context of CRNs, is motivated by, and can be seen as an analytic extension of, the global algebro-geometric viewpoint of CRN theory \cite{craciun2009toric,craciun2022disguised}.
Whereas the focus in previous articles was put on the global structure of the equilibrium manifold, the focus of this article is the analysis of the local properties.
Section \ref{sec:Riemannian} describes the Riemannian structure which can be thought of as encoding the linearization of the Legendre transform in the Riemannian metric tensor.
This structure is used in Section \ref{sec:Hessian_CRB} to prove the information geometric Cramer-Rao bound for CRN, which is the main technical result of this article.

\subsection{Riemannian structure} \label{sec:Riemannian}

The Legendre duality between the spaces $X$ and $Y$ is compatible with a Riemannian metric structure given by the Hessians of the respective convex functions $\phi$ and $\phi^*$ in the following sense:
First, the Legendre transform is an isometry, and second, it is natural to identify the cotangent bundle $T^*X \rightarrow X$ with the tangent bundle $TY \rightarrow Y$.
With this identification, the Riemannian metrics become the derivatives of the Legendre transformation.
In this sense, the Riemannian structure can be thought of as locally encoding the Legendre duality.
This structure restricts to the pair of subspaces $\V$ and $\Z$ and descends to the pair of parameter spaces $H$ and $\Q$.
This is now made precise.\\

The aim of the first part of the section is to discuss the geometric properties of Legendre dual spaces in general.
Therefore, let $X \subset \R^n$ and $Y \subset \RnT$ be any pair of Legendre dual spaces for now.
The strictly convex functions $\phi$ and $\phi^*$ define Riemannian metric tensors $g_X$ on $X$ and $g_Y$ on $Y$ via their Hessians as 
\begin{align*}
    g_X = \sum_{i,j=1}^n g_{X_{ij}} ~\dd x^i \otimes \dd x^j &:= \sum_{i,j=1}^n \frac{\partial^2 \phi(x)}{\partial x^i \partial x^j} ~\dd x^i \otimes \dd x^j, \\
    g_Y = \sum_{i,j=1}^n g_Y^{ij} ~\dd y_i \otimes \dd y_j &:= \sum_{i,j=1}^n \frac{\partial^2 \phi^*(y)}{\partial y_i \partial y_j} ~\dd y_i \otimes \dd y_j.
\end{align*}
The tensor $g_X$ gives an isomorphism between the tangent bundle and cotangent bundle over $X$ as
\begin{align}
\nonumber
    g_X: \qquad TX &\longrightarrow T^*X\\
\label{eq:g_X_iso}
    \sum_{i=1}^n v^i \frac{\partial}{\partial x^i} &\xmapsto{\quad} \sum_{i,j=1}^n v^i\frac{\partial^2 \phi(x)}{\partial x^i \partial x^j} \dd x^j,
\end{align}
and analogously for the tensor $g_Y$.
The Legendre transform between $X$ and $Y$ induces the bundle morphism
\begin{equation} \label{eq:diag_tangent}
    \begin{tikzcd}[nodes in empty cells]
    TX \ar[rr,shift left=.55ex,"D\LT_X"] \ar[d] & & TY \ar[ll,shift left=.55ex,"D\LT^{-1}_X"] \ar[d] \\
    X  \ar[rr,shift left=.55ex,"\LT_X"]  & &  
    Y  \ar[ll,shift left=.55ex,"\LT^{-1}_X"],
\end{tikzcd}
\end{equation}
where the derivative $D_x\LT_X$ at a point $x \in X$ is
\begin{align}
\nonumber
    D_x\LT_X: \qquad T_x X &\longrightarrow T_{y} Y \\
    \label{eq:tangent_map}
    \sum_{i=1}^n v^i \frac{\partial}{\partial x^i} &\xmapsto{\quad} \sum_{i,j=1}^n v^i \frac{\partial y_j}{\partial x^i} \frac{\partial }{\partial y_j} = \sum_{i,j=1}^n v^i \frac{\partial^2 \phi(x)}{\partial x^i \partial x^j} \frac{\partial }{\partial y_j}.
\end{align}
In the calculation, $y = \LT_X(x)$ and the expression (\ref{eq:LT_y}) for the $y_j$ was used.
Comparing (\ref{eq:tangent_map}) to the isomorphism (\ref{eq:g_X_iso}) yields the compatibility of the metric tensor $g_X$ and the bundle morphism $D\LT_X$ via the commutativity of the following diagram\footnote{In this diagram, it is understood that the bundles $TX$ and $T^*X$ are over the base $X$ and $TY$ over $Y$, and that the bundle morphisms are given with respect to the identity map $X \rightarrow X$ and $X \xrightarrow{\LT_X} Y$, respectively.}
\begin{equation} \label{eq:diag_compatibility}
    \begin{tikzcd}[nodes in empty cells]
    TX \ar[rr,"D\LT_X"] \ar[dr,swap,"g_X"] & & TY \\
    & T^*X  \ar[ur,swap,"\dd x^i \mapsto \frac{\partial}{\partial y_i}"] & 
\end{tikzcd}
\end{equation}
In this sense, the Riemannian structure encodes the infinitesimal geometry of the Legendre transformation.
The analogous statements hold for the tensor $g_Y$ and therefore it follows that the composition of morphisms
\begin{equation*}
    TX \xrightarrow{\quad g_X \quad } T^*X  \xrightarrow{\dd x^i \mapsto \frac{\partial}{\partial y_i}} TY \xrightarrow{\quad g_Y \quad } T^*Y \xrightarrow{\dd y_i \mapsto \frac{\partial}{\partial x^i}} TX
\end{equation*}
is the identity because it is equal to the composition $TX \xrightarrow{D\LT_X} TY \xrightarrow{D\LT^{-1}_X} TX$.
This shows that, for fixed $x$ and $y = \LT_X(x)$, the matrices $g_X(x)$ and $g_Y(y)$ are mutually inverse, i.e.,
\begin{equation} \label{eq:matrix_inverse}
    \sum_{j=1}^n [g_X(x)]_{ij} [g_Y(y)]^{jk} = \delta_i^k,
\end{equation}
where $\delta_i^k =1$ if $i = k$ and else $\delta_i^k = 0$.\\

\noindent These considerations lead to the following lemma, which plays a central role in the geometry of Legendre dual spaces:
\begin{lemma} \label{lemma:isometry}
    The Legendre duality
\begin{equation*}
    \begin{tikzcd}[nodes in empty cells]
    X  \ar[r,shift left=.55ex,"\LT_X"]  & 
    Y  \ar[l,shift left=.55ex,"\LT^{-1}_X"]
\end{tikzcd}
\end{equation*}
is an isometry with respect to the metric tensors $g_X$ and $g_Y$.
\end{lemma}
\begin{proof}
    It suffices to compute the pullback of the tensor $g_Y$ via $\LT_X$ using $g_{X_{ik}} = \frac{\partial y_i}{\partial x^k}$ and the relation (\ref{eq:matrix_inverse}):
    \begin{align*}
        \LT^*_X g_Y &= \sum_{i,j=1}^n g_Y^{ij} \sum_{k,l=1}^n \frac{\partial y_i}{\partial x^k}\frac{\partial y_j}{\partial x^l} ~\dd x^k \otimes \dd x^l = \sum_{i,j,k,l=1}^n g_Y^{ij} g_{X_{ik}} g_{X_{jl}} ~\dd x^k \otimes \dd x^l\\
        &=  \sum_{j,k,l=1}^n \delta^j_k g_{X_{jl}} ~\dd x^k \otimes \dd x^l = \sum_{k,l=1}^n g_{X_{kl}} ~\dd x^k \otimes \dd x^l = g_X.
    \end{align*}
\end{proof}

This concludes the general considerations on the geometry of Legendre dual spaces.
The lemma is now applied to the Diagram (\ref{eq:diag_comm}) to yield the final result of this subsection:
\begin{theorem} \label{theorem:Birch}
For the spaces described in Section \ref{sec:Hessian_setup}, the following diagram
    \begin{equation} \label{eq:diag_isom}
    \begin{tikzcd}[nodes in empty cells]
    (X,g_X,\phi) \ar[rr,shift left=.55ex,"\LT_X"] & & (Y,g_Y,\phi^*) \ar[ll,shift left=.55ex,"\LT^{-1}_X"] \\
    (\mathcal{V},\left.g_X\right|_{\mathcal{V}},\left.\phi\right|_{\mathcal{V}}) \arrow[u, hook] \ar[d, shift right=.55ex, swap, "U"] \ar[rr,shift left=.55ex,"\LT_X"]  & &  
    (\Z,\left.g_Y\right|_{\Z},\left.\phi^*\right|_{\Z}) \arrow[u, hook] \ar[ll,shift left=.55ex,"\LT^{-1}_X"]\\
    (H,\gH,\psi)  \ar[rr,shift left=.55ex,"\LT_H"] \ar[u, shift right=.55ex, swap, dotted, "\beta"] & & (\Q,\gl,\psi^*) \ar[ll,shift left=.55ex,"\LT^{-1}_H"] \ar[u, "f"].
\end{tikzcd}
\end{equation}
is a commutative diagram of isometries for the Riemannian metric tensors generated by the Hessians of the respective convex functions.
\end{theorem}
\begin{proof}
    The commutativity of the diagram follows from Proposition \ref{proposition:Birch} and the isometry in the rows from Lemma \ref{lemma:isometry}.
    The inclusions $\V \hookrightarrow X$ and $\Z \hookrightarrow Y$ with the restricted metrics $\left.g_X\right|_{\mathcal{V}}$ and $\left.g_Y\right|_{\Z}$ are automatically isometries.
    It remains to check the isometry between the second and the third row.
    Due to the commutativity of the diagram it suffices to do this for the right hand side, i.e., for the map $f$.
    An explicit computation shows that the Hessian metric $\gl$ is equal to the pullback $f^*g_Y$:
\begin{multline} \label{eq:gl}
    \gl = \sum_{i,j=1}^q \frac{\partial^2 \psi^*(\lambda)}{\partial \lambda_i \partial \lambda_j} \dd \lambda_i \otimes \dd \lambda_j = \sum_{i,j=1}^q \left[ \sum_{k,l=1}^n \frac{\partial^2 \phi^*(y)}{\partial y_k \partial y_l} \frac{\partial y_k}{\partial \lambda_i} \frac{\partial y_l}{\partial \lambda_j} + \right. \\  \left. + \sum_{k=1}^n \frac{\partial^2 y_k}{\partial \lambda_i \partial \lambda_j} \right]  \dd \lambda_i \otimes \dd \lambda_j =  \sum_{i,j=1}^q \sum_{k,l=1}^n \frac{\partial^2 \phi^*(y)}{\partial y_k \partial y_l} \frac{\partial y_k}{\partial \lambda_i} \frac{\partial y_l}{\partial \lambda_j}  \dd \lambda_i \otimes \dd \lambda_j = f^* g_Y
\end{multline}
    as $y = y^0 + \UU^*\lambda$ is linear in $\lambda$ and thus the terms $\frac{\partial^2 y_k}{\partial \lambda_i \partial \lambda_j}$ vanish.
\end{proof}
The Diagram (\ref{eq:diag_isom}) is used in the following section to derive an alternative expression for the metric $\left.g_X\right|_{\mathcal{V}}$ and to establish a Cramer-Rao bound for CRN by comparing it to the metric $g_X$.

\subsection{The information-geometric Cramer-Rao bound} \label{sec:Hessian_CRB}

In this section, all arguments are local at a given point $x \in \V$.
Let $\eta \in H$ and $\lambda \in \Q$ be its coordinates and $y = \LT(x) \in \Z$ its Legendre transform.
The arguments are not coordinate free but use the coordinate systems $H$ and $\Q$ for the spaces $\V$ and $\Z$, respectively.
Therefore, the Riemannian tensors and derivatives are expressed by their components in the $H$ and $\Q$ coordinates and in the $X$ and $Y$ coordinates.
For example, the relation (\ref{eq:matrix_inverse}), $g_X(x) = [g_Y(y = \LT_X(x))]^{-1}$ is abbreviated as $g_X = g_Y^{-1}$ and the component matrix $\gl(\lambda)$ of the metric $\gl$, given by $\gl(\lambda) = Ug_Y(y = f(\lambda))U^* $ according to (\ref{eq:gl}), is abbreviated as $\gl = Ug_YU^*$.\\

As the spaces $\Q$ and $H$ are Legendre dual, the relation (\ref{eq:matrix_inverse}) applies and it follows that the metric $\gH$ can be represented as 
\begin{equation*}
    \gH = \gl^{-1} =  [Ug_X^{-1}U^*]^{-1}.
\end{equation*}
Theorem \ref{theorem:Birch} now states that $\left.g_X\right|_{\mathcal{V}}$ is the pullback of $\gH$ to $\mathcal{V}$.
Therefore, it can be written explicitly as
\begin{equation} \label{eq:isom}
    \left.g_X\right|_{\mathcal{V}} = U^* \gH U = U^* [U g_X^{-1} U^*]^{-1} U. 
\end{equation}
This equality of metrics holds on $T_x \mathcal{V}$ by construction.
However, the $n \times n$ matrix $U^* \gH U$ can be used to define a (pseudo)metric on the whole $T_xX$.
To this end, a characterization of the complement of $T_x \mathcal{V}$ in $T_xX$ is needed.
\begin{lemma} \label{lem:comp}
    The $g_X$-orthogonal complement $(T_x \mathcal{V})^{\perp} \subset T_xX$ is given by $\Ker[U]$.
\end{lemma}
\begin{proof}
   Let $v \in T_xX$.
   A vector $w \in T_x\mathcal{V}$ can be expressed as the pushforward of some tangent vector $\overline{w} \in T_{\lambda} \Q$, i.e., $w = \frac{\partial x}{\partial y} \frac{\partial y}{\partial \lambda} \overline{w}$ with $\lambda = f^{-1} \circ \LT_X(x)$.
   The inner product between $v$ and $w$ is given by
\begin{align*}
    \langle v, w \rangle_{g_X} &= \sum_{i,j=1}^n v^i \frac{\partial^2 \phi(x)}{\partial x^i \partial x^j} w^j = \sum_{i,j=1}^n v^i \frac{\partial y_i}{\partial x^j} w^j = \sum_{i,j,k=1}^n \sum_{l=1}^q v^i \frac{\partial y_i}{\partial x^j} \frac{\partial x^j}{\partial y_k} \frac{\partial y_k}{\partial \lambda_l} \bar{w}_l \\
    &= \sum_{i=1}^n \sum_{l=1}^q v^i \frac{\partial y_i}{\partial \lambda_l} \bar{w}_l = \sum_{i=1}^n \sum_{l=1}^q v^i u_i^l \bar{w}_l = \langle Uv, \bar{w} \rangle,
\end{align*}
   where $g_X = \frac{\partial^2 \phi(x)}{\partial x \partial x} = \frac{\partial y}{\partial x}$ and $\frac{\partial y}{\partial \lambda} = U^*$ was used.
   This inner product vanishes if $v \in \Ker[U]$.
   Thus $\Ker[U] \subset (T_x \mathcal{V})^{\perp}$ and the equality follows by comparing the dimensions of both $\Ker[U]$ and $(T_x \mathcal{V})^{\perp}$, which are equal to $n - q$.
\end{proof}
This lemma, together with Theorem \ref{theorem:Birch}, yields a characterization of $U^* \gH U$.
\begin{lemma} \label{lemma:Hessian_metric}
The matrix $U^* \gH U$ defines a pseudo-Riemannian metric on $X$.
The metric agrees with $g_X$ on $T_x \mathcal{V}$ and vanishes on its $g_X$-orthogonal complement.
\end{lemma}
\begin{proof}
    Note that it is required that $x \in \mathcal{V}$ for the construction given thus far to hold.
    However, the extension to the whole space $X$ is straightforward:
    By varying the base point $y^0$ in the definition of the map $f$, one obtains a foliation of $Y$ by affine spaces $\Z$ and a thereby a foliation of $X$ by spaces $\mathcal{V}$ from Legendre duality.
    Thus every point $x \in X$ lies on a unique $\mathcal{V}$, and this is the space to be considered in the statement for arbitrary $x \in X$.
    The agreement of $g_X$ and $U^* \gH U$ on $T_x \mathcal{V}$ is ensured by the isometry of Diagram \ref{eq:diag_isom}.
    By Lemma \ref{lem:comp}, the $g_X$-orthogonal complement of $T_x \mathcal{V}$ is $\Ker[U]$, where $U^* \gH U$ vanishes.
\end{proof}

Lemma \ref{lemma:Hessian_metric} can be reformulated as the Hessian geometric version of the Cramer-Rao bound for CRN:
\begin{theorem}[Hessian geometric Cramer-Rao bound] \label{thm:Hessian_CRB}
The matrix inequality
\begin{equation} \label{eq:CRB}
    g_X \geq U^* \gH U
\end{equation}
holds, i.e., the difference between the matrices $g_X$ and $U^* \gH U$ is positive semidefinite.
The matrix $\gH$ is explicitly given by $\gH = \gl^{-1} = [\UU g_X^{-1} \UU^*]^{-1}$.
\end{theorem}
Finally, the matrix $U^* \gH U$ can be characterized as follows:
\begin{corollary} \label{corr:CRB}
Let
\begin{equation*}
    \pi: T_x X \rightarrow T_x \V
\end{equation*}
be the $g_X$-orthogonal projection.
Then the matrix elements of $U^* \gH U$ are given by
\begin{equation}
   [\UU^* \gH \UU]_{ij} = \left\langle \pi\left(\frac{\partial}{\partial x^i} \right), \pi\left( \frac{\partial}{\partial x^j} \right) \right\rangle_{g_X}.
\end{equation}
\end{corollary}
\begin{proof}
    This follows directly from Lemma \ref{lemma:Hessian_metric}.
\end{proof}

In this section, the Riemannian geometry of quasi-equilibrium manifolds $\V$ has been discussed, with the Cramer-Rao bound as the main result.
This bound might be of interest for CRN theory in future applications by itself.
In fact, it is known that the metric $g_X$ controls the dissipation during the driving of the CRN \cite{loutchko2022riemannian} and is thus of thermodynamical importance.
It is worth noting that this multivariate Cramer-Rao bound is somehow orthogonal to the scalar bound derived in \cite{yoshimura2021}.
In the following section, the results are applied to the characterization of absolute sensitivity in quasi-equilibrium manifolds.

\section{Absolute Sensitivity} \label{sec:abs_sens}

\subsection{Motivation and definition}
As pointed out in Section \ref{sec:sens_intro}, classical sensitivity is defined as the derivative of the section $\gamma: \bX \rightarrow X$.
As such, the classical sensitivity matrix $D_{\eta} \beta = \frac{\partial x}{\partial \eta}$ at a point $x \in V$ with coordinates $\eta \in H$ depends on the choice of the coordinate system $H$, which is equivalent to choosing a basis for $\Ker[S^*]$.
In \cite{shinar2009}, Shinar, Alon, and Feinberg posed the question of how to define quantities that serve a similar purpose to the sensitivity matrix but are canonical, i.e., independent of the choice of a basis for $\Ker[S^*]$.
This is achieved by absolute sensitivity \cite{loutchko2024cramer}.

Instead of the derivative of $\gamma$, one can consider the derivative of the composition 
\begin{equation*}
    X \xlongrightarrow{p} \bX \xlongrightarrow{\gamma} X,
\end{equation*}
which, at any point $x \in X$, can be expressed in the canonical coordinates of the concentration space, and is independent of the choice of a coordinate system for $\bX$, i.e., independent of the choice of a basis for $\Ker[S^*]$.
Moreover, one verifies that the Jacobian matrix $D_x [\gamma \circ p]$ is the same for all points in the stoichiometric polytope $P(x)$.
Therefore, it is enough to understand this derivative at equilibrium points $x \in \V$.
This leads to the definition:

\begin{definition} \label{def:abs_sens}
    For a fixed base point $y^0 \in Y$ and the corresponding $\V$, absolute sensitivity $A$ is defined as the bundle morphism
\begin{equation} \label{eq:def_abs_sens_intro}
    \begin{tikzcd}[column sep=small]
    A \arrow[r,phantom,":" description]& TX \ar[rr,"Dp"] \ar[d] & & T\bX \ar[rr,"D\gamma"] \ar[d] & & TX \ar[d] \\
    & \V  \ar[rr,"p"]  & &  
    \bX \ar[rr,"\gamma"] & & \V,
\end{tikzcd}
\end{equation}
where $TX \rightarrow \V$ is the pullback of $TX \rightarrow X$ along the inclusion $\V \hookrightarrow X$. 
\end{definition}
Some remarks are in order to clarify and elaborate on the definition:
\begin{remark}
    As $\gamma$ maps to $\V$, its derivative $D\gamma$ will have its image in $T\V$ inside $TX$.
    Moreover, by definition, the restriction of $A$ to $T\V$ is the identity.
    Therefore, $A$ is a projection operator.
    In fact, it is the $g_X$-orthogonal projection operator which is the content of Theorem \ref{thm:abs_sens_final}.
\end{remark}
\begin{remark} \label{rmk:extend_to_X}
    Varying the basepoint $y^0$ in the definition of $\V$ gives a fibration of $X$ by equilibrium manifolds $\V$ and thus the definition of absolute sensitivity can be extended from $TX \rightarrow \V$ to $TX \rightarrow X$.
    This is the same argument as in the proof of Lemma \ref{lemma:Hessian_metric}.
\end{remark}
\begin{remark}
    When using the coordinate system $H$ for $\bX$, the absolute sensitivity is the bundle morphism $A:TX \xrightarrow{DU} TH \xrightarrow{D\beta} TX$.
    In this coordinate system, the absolute sensitivity $A(x)$ at a point $x \in \V$ is given by the matrix product
    \begin{equation*}
        A(x) = [D_{Ux} \beta] [D_x U],
    \end{equation*}
    with the matrix elements
    \begin{align*} \label{eq:aij}
    A(x)^i_j = \sum_{k=1}^q \frac{\partial x^j}{\partial \eta^k} u^k_{j}.
    \end{align*}
    This is the definition used in \cite{loutchko2024cramer}.
    In applications, the diagonal elements of $A(x)$ play an important role as they are numerical invariants for the sensitivity of the respective chemical.
    This motivates the definition:
\end{remark}
\begin{definition}
    The matrix element $A(x)^i_j$ is called the {\it absolute sensitivity of $X_i$ with respect to $X_j$} at a point $x \in \V$.
    
    The {\it absolute sensitivity $\alpha_i(x)$ of the chemical $X_i$ at a point $x \in \V$ is defined as} 
    \begin{equation*}
        \alpha_i := A(x)_i^i.
    \end{equation*}  
\end{definition}
\begin{remark}
    It is worth discussing the intuition behind the definition:
    Fix a point $x \in \V$ and let $\delta x^j \in T_xX$ be an infinitesimal perturbation in the $x^j$-direction, i.e., a change of the concentration of a chemical $X_j$.
    This changes the vector of conserved quantities by $D_x U \delta x^j$, leading to a new equilibrium state with the adjusted concentrations given by $D_{Ux} \beta D_x U \delta x^j$, which is just $A(x) \delta x^j$.
    This justifies calling the matrix element $A(x)^i_j$ the absolute sensitivity of $X_i$ with respect to $X_j$, as it quantifies the fraction of concentration change in $X_j$ which distributes to $X_i$.
    Further explanations can be found in \cite{loutchko2024cramer}.
\end{remark}

The following theorem shows the advantage of absolute sensitivity over the sensitivity matrix $\frac{\partial x}{\partial \eta}$:
\begin{theorem}[\cite{loutchko2024cramer}, Theorem 3.4.]
The matrix of absolute sensitivities $\A(x)$ is independent of the choice of a basis of $\Ker[S^*]$.
Moreover, the equality
\begin{align*}
   \mathrm{Tr}[\A(x)] = \sum_{i=1}^n \alpha_i(x) = q 
\end{align*}
holds, whereby $q = \dim \Ker[S^*]$.
\end{theorem}
\noindent For quasi-thermostatic CRN (cf. Example \ref{ex:quasiTS} for the definition), a linear algebraic characterization of the matrix of absolute sensitivities has been derived in \cite{loutchko2024cramer} by using a Cramer-Rao bound.
In the next section a generalization to arbitrary quasi-equilibrium manifolds is provided.

\subsection{Absolute sensitivity of quasi-equilibrium manifolds} \label{sec:abs_sens_comp}

Even for quasi-thermostatic CRNs (cf. Example \ref{ex:quasiTS}), there is no analytical expression for the section $\beta$ and therefore direct computations are not possible based purely on the Definition \ref{def:abs_sens}.
This difficulty can be circumvented by using the Diagram (\ref{eq:diag_isom}) and the Cramer-Rao bound which results in a linear algebraic characterization of absolute sensitivity.
The main result is that $A$ is the $g_X$-orthogonal projection $\pi: T X \rightarrow T \V$.
Whereas the projection property follows from the definition, the $g_X$-orthogonality requires the geometrical results from Section \ref{sec:information_geometry}.
The linear algebraic characterization is needed for explicit analytical computations and thus for the understanding of concentration robustness and hypersensitivity phenomena in specific CRN (see Section \ref{sec:IDH} for an example).

The commutative Diagram (\ref{eq:diag_comm}) is used to represent $A$ explicitly as a product of known derivatives.
Then, by using the compatibility between Hessian metrics and derivatives via Diagram (\ref{eq:diag_compatibility}), this representation is related to the right hand side of the Cramer-Rao bound.\\
    
The absolute sensitivity is given by $TX \xrightarrow{DU} TH \xrightarrow{D\beta} TX$, which is resolved, using the commutativity of the Diagram (\ref{eq:diag_comm}), as
\begin{equation*}
    \underbrace{\T X \xrightarrow{\quad DU \quad}  TH \xrightarrow{\quad D\LT_H \quad} T\Q \xrightarrow{\quad Df \quad} T\Z}_{U^* \gH U} \xrightarrow{\quad D\LT^{-1}_X \quad} T\V \xhookrightarrow{\qquad} TX,
\end{equation*}
i.e., by the expression
\begin{equation*}
    A = D \LT_X^{-1} \circ Df \circ D\LT_H \circ DU.
\end{equation*}
Using the correspondence in Diagram (\ref{eq:diag_compatibility}), the concatenation of derivatives $Df \circ D\LT_H \circ DU$ corresponds to $U^* \gH U$ whereas $D \LT_X^{-1}$ corresponds to $g_Y$, giving
\begin{equation*}
    A = g_Y U^* \gH U.
\end{equation*}
Finally, using the relations (\ref{eq:gl}) and (\ref{eq:matrix_inverse}) to resolve $\gH(\eta)$ as $[U g_Y(y) U^*]^{-1}$ leads to the expression
\begin{equation} \label{eq:A_full}
    A(x) = g_Y(y)U^* [U g_Y(y) U^*]^{-1} U.
\end{equation}
The formula (\ref{eq:A_full}) shows that $A$ is idempotent:
\begin{equation*}
    A(x)^2 = g_Y(y)U^* [U g_Y(y) U^*]^{-1} Ug_Y(y)U^* [U g_Y(y) U^*]^{-1} U = A(x)
\end{equation*}
and that, therefore, $A$ is a projection operator from $TX$ to $T \V$:
\begin{theorem} \label{thm:abs_sens_final}
    The absolute sensitivity $A$ is the $g_X$-orthogonal projection $\pi: TX \rightarrow T\V$.
    For a fixed $x \in \V$, its components are given explicitly as
    \begin{equation} \label{eq:A_linear_alg}
        A(x)^i_j = \left\langle \dd x^i , \pi\left( \frac{\partial}{\partial x^j} \right) \right\rangle
    \end{equation}    
    in the canonical $X$ coordinates.
\end{theorem}
\begin{proof}
    Let $\pi$ be given by $\pi \left( \sum_{j=1}^n v^j \frac{\partial }{\partial x^j} \right) = \sum_{j,l = 1}^n v^j \pi^l_j \frac{\partial}{\partial x^l}$.
    The representation $A = g_Y [U^* \gH U]$, together with Corollary \ref{corr:CRB} and $\left\langle \pi\left(\frac{\partial}{\partial x^k} \right), \pi\left( \frac{\partial}{\partial x^j} \right) \right\rangle_{g_X} = \left\langle \frac{\partial}{\partial x^k}, \pi\left( \frac{\partial}{\partial x^j} \right) \right\rangle_{g_X}$, yields 
\begin{multline*}
      A(x)^i_j = \sum_{k=1}^n g_Y^{ik}  \left\langle \pi\left(\frac{\partial}{\partial x^k} \right), \pi\left( \frac{\partial}{\partial x^j} \right) \right\rangle_{g_X} = \\
      = \sum_{k=1}^n g_Y^{ik}  \left\langle \frac{\partial}{\partial x^k} , \sum_{l=1}^n \pi_j^l  \frac{\partial}{\partial x^l} \right\rangle_{g_X}
      = \sum_{k,l=1}^n g_Y^{ik} g_{X_{kl}} \pi_j^l = \sum_{l=1}^n \delta^i_l \pi^l_j = \pi^i_j,
\end{multline*}
where (\ref{eq:matrix_inverse}) was used and which is the desired relation between $A$ and $\pi$.
Finally, the tautology $\pi^i_j = \left\langle \dd x^i , \pi\left( \frac{\partial}{\partial x^j} \right) \right\rangle$ yields the characterization (\ref{eq:A_linear_alg}) of the matrix elements of $A$.
\end{proof}

This interpretation of the matrix of absolute sensitivities allows for a global viewpoint over $X$ as follows:
\begin{definition}
    The global absolute sensitivity is the bundle morphism
\begin{equation*}
    A: TX \rightarrow TX
\end{equation*}
which is locally given by the the $g_X$-orthogonal projection to the subbundle 
\begin{equation*}
    T \V := D\LT^{-1}_X \Img[U^*],
\end{equation*}
where $\Img[U^*]$ is the corresponding subbundle of $TY$.
\end{definition}
The crucial point of this definition is that {\it a priori} the absolute sensitivity has been defined only at points $x \in \V$ but the above definition is given globally on $X$.
The reason for this is that, by varying the base point $y^0$ in the definition of the map $f$ in (\ref{eq:f_def}), one obtains a foliation of $X$ by spaces $\V$, cf. Remark \ref{rmk:extend_to_X}.\\

\noindent 
To close this section, observe, by using formula (\ref{eq:A_full}), that the matrix $A(x)g_Y(y)$ is symmetric.
Taking into account that $g_Y(y)$ is itself symmetric, this yields the equality $A(x)g_Y(y) = g_Y(y) A^*(x)$ and the following lemma.
\begin{lemma} \label{lemma:A_perp}
    The following diagram of bundle morphisms commutes
\begin{equation*}
    \begin{tikzcd}[nodes in empty cells]
    T^*X \ar[d, "A^*"] \ar[rr,"\dd x^i \mapsto \frac{\partial}{\partial y_i}"] & & TY \ar[rr, "g_Y"] & & TX \ar[d, "A"]\\
    T^*X  \ar[rr,"\dd x^i \mapsto \frac{\partial}{\partial y_i}"]  & & TY \ar[rr, "g_Y"] & & TX.
\end{tikzcd}
\end{equation*}
Moreover, let $A^{\perp}: TX \rightarrow TX$ be the bundle morphism given by the $g_X$-orthogonal projection
\begin{equation*}
    A^{\perp}: T X \rightarrow (T \V)^{\perp}.
\end{equation*}
Then $A^{\perp}$ and its adjoint $(A^{\perp})^*$ satisfy the analogous commutative diagram of bundle morphisms.
\end{lemma}
\begin{proof}
    The commutativity of the shown diagram is equivalent to $A(x)g_Y(y) = g_Y(y) A^*(x)$.
    The morphism $A^{\perp}$ is given by $A^{\perp} = I - A$, where $I: TX \rightarrow TX$ is the identity.
    The commutativity $A^{\perp}g_Y = g_Y (A^{\perp})^*$ then follows from $g_Y I = g_Y (A + A^{\perp}) = I g_Y = (A + A^{\perp})^* g_Y$ by using $Ag_Y = g_Y A^*$.
\end{proof}

The orthogonal projection $A^{\perp}$ plays an important role for the first order corrections to the ideal thermodynamic behavior, as is shown in the following section.

\section{Beyond quasi-thermostatic CRNs: Deformations of the metric and first order correction results} \label{sec:first_order}

As a novel application of the information geometric framework, the effects of non-ideal thermodynamic behavior on the absolute sensitivity can be analyzed.
Geometrically, this corresponds to first order deformations of the steady state manifold $\V^{\mathrm{id}}$ from Example \ref{ex:quasiTS}.
Therefore, after recalling the theory for $\V^{\mathrm{id}}$, which was worked out in \cite{loutchko2024cramer}, the first order correction of a deformation of the metric $g_X^{\mathrm{id}}$ are analyzed. 

\subsection{Quasi-thermostatic CRN} 
The results from \cite{loutchko2024cramer} for quasi-thermostatic CRN can be recovered for the convex function $\phi^{\mathrm{id}}(x) = \sum_{i =1}^n [\mu_i^0 x^i + x^i \log x^i - x^i]$, cf. Example \ref{ex:quasiTS}.
The Hessian metric $g^{\mathrm{id}}_X$ is given by $g^{\mathrm{id}}_X = \sum_{i=1}^n \frac{1}{x^i} \dd x^i \otimes \dd x^i$, i.e., it is represented by the diagonal matrix $g^{\mathrm{id}}_X(x) = \XX$.
The inverse matrix is $g^{\mathrm{id}}_Y(y) = \X$, and the Cramer-Rao bound, which was derived by standard linear algebra in \cite{loutchko2024cramer}, reads 
\begin{equation*}
    \XX \geq U^* [U \X U^*]^{-1} U.
\end{equation*}
Then Theorem \ref{thm:abs_sens_final} specializes to the expression 
\begin{equation*}
    A(x)^i_j = x^i \left\langle \pi\left(\frac{\partial}{\partial x^i}\right),\pi\left(\frac{\partial}{\partial x^j}\right)\right\rangle_{g^{\mathrm{id}}_X },
\end{equation*}
which was proven in \cite{loutchko2024cramer}.
Moreover, this leads to the following corollary.
\begin{corollary}[\cite{loutchko2024cramer}, Corollary 4.5.] \label{corr:alpha_in_0_1}
For quasi-thermostatic CRN, the absolute sensitivities $\ai$ satisfy
\begin{equation*}
        \ai \in [0,1].
\end{equation*}    
\end{corollary}
However, the analogous statement of the corollary in the more general setup of this article is not expected to hold, which can lead to interesting behavior of the respective CRN.

\subsection{First order corrections} \label{sec:first_order_comp}

The geometry of first order corrections to the metric $g_X$ can be analyzed algebraically and the results have a geometrical interpretation.
Let the metric tensor $g_X$ be given by
\begin{equation} \label{eq:gX_deform}
    g_{X_{kl}} = \frac{\partial^2 \phi(x)}{\partial x_l \partial x_k} = {g_X^{\mathrm{id}}}_{kl} + c_{kl},
\end{equation}
where $g^{\mathrm{id}}_X(x) = \XX$ is the metric corresponding to the ideal behavior of quasi-thermostatic CRN and the $c_{kl}$ are correction terms.
The correction terms are required to be symmetric, i.e., $c_{kl} = c_{lk}$ for all $l,k$, and the matrix $g_X(x)$ must be positive definite.
The potential with quadratic correction terms (Example \ref{ex:quadratic}) and the van der Waals potential (Example \ref{ex:vdW}) give this type of modification of the metric.

From now on, the calculations are carried out to first order in the $c_{kl}$\footnote{To work in first order in the $c_{kl}$ means to work in the ring
\begin{equation*}
     \faktor{\mathbb{R}(x^1,\dotsc,x^n)[c_{11},c_{12},\dotsc,c_{nn}]}{\left(c_{ij}c_{kl}\right)_{i,j,k,l=1,\dotsc,n}}.
\end{equation*}
This is the ring generated, over the field $\mathbb{R}(x^1,\dotsc,x^n)$ of rational functions in $n$ indeterminates, by $c_{11},c_{12},\dotsc,c_{nn}$ such that $c_{ij}c_{kl} = 0$.
In this ring, for $A \in \mathrm{GL}_n (\mathbb{R}(x^1,\dotsc,x^n))$ and $B \in \mathrm{M}_n ((c_{11},c_{12},\dotsc,c_{nn}) \cdot \mathbb{R}(x^1,\cdots,x^n))$, the matrix inversion formula $[A + B]^{-1} = A^{-1} - A^{-1}B A^{-1}$ holds true.
This is used for the symbolic computations in SageMath in the Supplementary Material.
}.
Let $C$ be the $n \times n$ matrix of correction terms, i.e., $C_{ij} = c_{ij}$ and $g_X = g_X^{\mathrm{id}} + C$.
In the following calculations, it is understood that one works with the local representations $g_X(x)$, $g_Y(y)$, etc. but for notational convenience the tensor symbols $g_X$, $g_Y$, etc. are written, cf. Section \ref{sec:Hessian_CRB}.
The metric $g_Y$ is given by
\begin{equation*}
    g_Y = \left[g_X^{\mathrm{id}} + C \right]^{-1} =  g_Y^{\mathrm{id}} - g_Y^{\mathrm{id}} C g_Y^{\mathrm{id}}
\end{equation*}
and one also obtains 
\begin{align*}
    \gH &= [Ug_YU^*]^{-1} = [Ug_Y^{\mathrm{id}}U^*]^{-1} + [Ug_Y^{\mathrm{id}}U^*]^{-1} U g_Y^{\mathrm{id}} C g_Y^{\mathrm{id}} U^* [Ug_Y^{\mathrm{id}}U^*]^{-1} \\
        &= \gH^{\mathrm{id}}+ \gH^{\mathrm{id}} U g_Y^{\mathrm{id}} C g_Y^{\mathrm{id}} U^* \gH^{\mathrm{id}},
\end{align*}
where $\gH^{\mathrm{id}} = [Ug_Y^{\mathrm{id}}U^*]^{-1}$.
This yields, using the expression (\ref{eq:A_full}), for the matrix of absolute sensitivities
\begin{align*}
    A &= g_Y U^* \gH U = \left[g_Y^{\mathrm{id}} - g_Y^{\mathrm{id}} C g_Y^{\mathrm{id}}\right] U^* \left[ \gH^{\mathrm{id}}+ \gH^{\mathrm{id}} U g_Y^{\mathrm{id}} C g_Y^{\mathrm{id}} U^* \gH^{\mathrm{id}} \right] U \\
    &= A^{\mathrm{id}} + A^{\mathrm{id}} g_Y^{\mathrm{id}} C A^{\mathrm{id}} - g_Y^{\mathrm{id}} C A^{\mathrm{id}} \\
    &= \left[ I - [I - A^{\mathrm{id}}]  g_Y^{\mathrm{id}} C \right] A^{\mathrm{id}} \\
    &= \left[ I - (A^{\perp})^{\mathrm{id}} g_Y^{\mathrm{id}} C \right] A^{\mathrm{id}} ,
\end{align*}
where $A^{\mathrm{id}} = g_Y^{\mathrm{id}} U^* \gH^{\mathrm{id}} U$ is the absolute sensitivity for quasi-thermostatic CRN and $(A^{\perp})^{\mathrm{id}}$ represents the $g_X^{\mathrm{id}}$-orthogonal projection to $(T_x \mathcal{V}^{\mathrm{id}})^{\perp}$, cf. Lemma \ref{lemma:A_perp}.
The first order correction to $A$ is thus given by $(A^{\perp})^{\mathrm{id}} g_Y^{\mathrm{id}} C A^{\mathrm{id}}$.
This has the geometrical interpretation as a concatenation of bundle morphisms over $\V$ as
\begin{equation*}
    \T X \xrightarrow{\quad A^{\mathrm{id}}  \quad}  TX \xrightarrow{\quad C \quad} TY \xrightarrow{\quad g_Y^{\mathrm{id}} \quad} TX \xrightarrow{\quad (A^{\perp})^{\mathrm{id}}  \quad} TX,
\end{equation*}
where the tensor $C = g_X - g_X^{\mathrm{id}}$ is interpreted as a bundle morphism between $TX$ and $TY$ according to (\ref{eq:diag_compatibility}).
The projection $(A^{\perp})^{\mathrm{id}}$ kills the space $T \V^{\mathrm{id}}$, which has the preimage $\Img[U^*] \subset TY$.
Thus, geometrically, the impact of the correction from $A^{\mathrm{id}}$ to $A$ is determined by how far $C$ takes $\Img[A^{\mathrm{id}} ] = T \V^{\mathrm{id}}$ away from $\Img[U^*] \subset TY$.
Using this geometrical interpretation, and taking into account the requirement for $C$ to be symmetric, one can derive criteria for the vanishing of the first order corrections:
\begin{lemma} \label{lemma:A_vanishing_corr}
    If, for each $x \in \V$, there exists a $n \times m$ matrix $M(x)$ with columns in $\Img[U^*]$ such that $C(x) = M(x)M^T(x)$, then the matrix of absolute sensitivities $A$ coincides with the matrix of absolute sensitivities for the corresponding quasi-thermostatic CRN, i.e.,
    \begin{equation*}
        A = A^{\mathrm{id}},
    \end{equation*}
    to first order in the $C(x)_{kl}$.
\end{lemma}
\begin{proof}
This follows from the preceding discussion.
The inclusion $\Img[C] \subset \Img[U^*]$ holds by assumption.
From $\Ker[S^*] = \Img[U^*]$, it follows that $(A^{\perp})^{\mathrm{id}}$ kills $g_Y^{\mathrm{id}} C$ because it is contained in $T \V^{\mathrm{id}}$, i.e, $(A^{\perp})^{\mathrm{id}} g_Y^{\mathrm{id}} C = 0$.
\end{proof}
\noindent One assumption one can make in calculations is that all $C_{ij}(x)$ are equal to some function $c(x)$.
In this case, $\Img[C]$ is spanned by the vector $\mathbb{1}_{TY} := \sum_{i=1}^n \frac{\partial }{\partial y_i}$ and the impact of the correction is governed by the position of $\mathbb{1}_{TY}$ relative to $\Img[U^*]$.
Lemma \ref{lemma:A_vanishing_corr} implies the following Corollary:
\begin{corollary} \label{corr:first_order}
If the total concentration $\sum_{i=1}^n x^i$ is conserved, i.e., if $$\sum_{i=1}^n e^i \in \Ker[S^*],$$ and if $C_{ij}(x) = c(x)$ for some function $c(x)$, then the matrix of absolute sensitivities $A$ coincides with the matrix of absolute sensitivities for the corresponding quasi-thermostatic CRN, i.e.,
    \begin{equation*}
        A = A^{\mathrm{id}},
    \end{equation*}
    to first order in $c(x)$.
\end{corollary} 

\subsection{Quadratic corrections and the van der Walls potential} \label{subsec:quadratic_vdW}

The potential with quadratic correction terms (\ref{eq:phi_quad}) from Example \ref{ex:quadratic} gives the Riemannian metric $g_{X_{kl}} =  g^{\mathrm{id}}_{X_{kl}} + c_{kl}$ with $c_{kl} = a_{kl} + a_{lk}$.
For weak interactions (i.e., when the $c_{kl}$ are small compared to the terms $\frac{1}{x_i}$ and to 1), the above results on first order corrections hold.\\

An explicit calculation with the van der Waals potential from formula (\ref{eq:phi_vdW}) yields the correction terms $c_{kl}$ as
\begin{equation*}
    c_{kl} = - \sum_{i=1}^n \frac{x^i b^i_k b^i_l}{\left(1 - \sum_{j=1}^n b^i_j x^j \right)^2} + \frac{b^l_k}{1 - \sum_{j=1}^n b^l_j x^j } + \frac{b^k_l}{1 - \sum_{j=1}^n b^k_j x^j }  - [a_{lk} + a_{kl}].
\end{equation*}
Indeed, for small constants $a_{ij}$ and $b^i_j$, the $c_{kl}$ can be approximated by $c_{kl} \approx [b^l_k + b^k_l] - [a_{lk} + a_{kl}]$ and the first order calculations are valid.\\

When the correction terms $c_{kl}$ are not small, it is difficult to derive closed formulae for the absolute sensitivity.
However, numerical results are easily accessible and provide interesting insights into the behavior of CRNs with non-ideal thermodynamics.
This is explored in the next section.

\section{Example: The core module of the IDHKP-IDH glyoxylate bypass regulation system} \label{sec:IDH}

To give a concrete example application, the effects of non-ideal behavior are analyzed for core module of the IDHKP-IDH (IDH = isocitrate dehydrogenase, KP = kinase-phosphatase) glyoxylate bypass regulation system.    This CRN is of interest as it is known to obey approximate concentration robustness in the IDH enzyme \cite{laporte1985compensatory}, and has been mathematically proven to posses absolute concentration robustness under mass action kinetics under certain irreversibility conditions \cite{shinar2009robustness,shinar2010structural}.
Detailed calculations and the SageMath code \cite{sagemath} is available in the Supplementary Material.

The corresponding CRN is shown in Scheme (\ref{eq:ACR}), whereby I is the IDH enzyme, $\textrm{I}_{\textrm{p}}$ is its phosphorylated form, and E is the bifunctional enzyme IDH kinase-phosphatase.
\begin{equation} \label{eq:ACR}
\begin{tikzcd}
    \textrm{E} + \textrm{I}_{\textrm{p}}  \ar[r, rightharpoonup, shift left=.35ex,"k_1^+"] & \textrm{EI}_{\textrm{p}} \ar[l,shift left=.35ex, rightharpoonup,"k_1^-"] \ar[r, rightharpoonup, shift left=.35ex,"k_2^+"] & \textrm{E} + \textrm{I} \ar[l,shift left=.35ex, rightharpoonup,"k_2^-"] \\
    \textrm{EI}_{\textrm{p}} + \textrm{I} \ar[r, rightharpoonup, shift left=.35ex,"k_3^+"] & \textrm{EI}_{\textrm{p}}\textrm{I} \ar[l,shift left=.35ex, rightharpoonup,"k_3^-"] \ar[r, rightharpoonup, shift left=.35ex,"k_4^+"] & \textrm{EI}_{\textrm{p}} + \textrm{I}_{\textrm{p}}. \ar[l,shift left=.35ex, rightharpoonup,"k_4^-"]
\end{tikzcd}
\end{equation}
In following, the notation $X_1 = \textrm{E}, X_2 = \textrm{I}_{\textrm{p}}, X_3 = \textrm{EI}_{\textrm{p}}, X_4 = \textrm{I}, X_5 = \textrm{EI}_{\textrm{p}}\textrm{I}$ for the chemicals and $x^i, i =1,\dotsc,5$ for the respective concentration values will be used.
In \cite{shinar2009robustness,shinar2010structural} it was shown that if the CRN obeys mass action kinetics and the rate constants $k_2^-$ and $k_4^-$ are zero then the concentration robustness in the IDH enzyme $X_4$ holds exactly.
In \cite{loutchko2024cramer}, Section 5, the absolute sensitivity $\alpha_4^{\mathrm{id}}$ was calculated for the quasi-thermostatic case, i.e., for $\V$ resulting from the convex function $\phi^{\mathrm{id}}$, which corresponds to mass action kinetics with detailed balancing and requires nonvanishing $k_2^-$ and $k_4^-$:
\begin{equation} \label{eq:a4}
    \alpha_4^{\mathrm{id}} = \frac{1}{1+r},
\end{equation}
where $r$ is the ratio
\begin{equation*}
    r = \frac{\left(x_2 + x_5 \right) \left(x_1 + x_3 \right) + x_1 \left( x_3 + 3x_5\right) +x_2x_5}{x_4 \left(x_1 + x_3 + x_5 \right)}.
\end{equation*}
This expression allows to read off regions of high and low sensitivity, cf. \cite{loutchko2024cramer}, Section 5.

\subsection{First order corrections}
Assuming that $C(x)_{ij} = c(x)$ for some function $c(x)$, one obtains
\begin{equation} \label{eq:alpha_4}
    \alpha_4 = [1 - c(x) \beta]\alpha_4^{\mathrm{id}}
\end{equation}
for the first order correction, with
\begin{equation*}
    \beta = \frac{(x_2 + x_4) (x_1x_3 + 4 x_1x_5 + x_3 x_5) }{(x_1+x_3+x_5) (x_1 + x_2 + x_4 + x_5)- (x_5-x_1)^2}.
\end{equation*}
By expanding the denominator, one verifies that $\beta >0$.
This allows to discuss the influence of the interaction terms from Section \ref{subsec:quadratic_vdW} on the absolute sensitivity:

For the case of quadratic interaction terms $a_{ij}x^ix^j$, the condition $C(x)_{ij} = c(x)$ requires all $a_{ij}$ to be equal, i.e., $a_{ij} = a$ for some $a$.
Then $c(x) = 2a$ and repulsive interactions, i.e., $a >0$, decrease the sensitivity $\alpha_4$, whereas attractive interactions, i.e., $a <0$, increase the sensitivity $\alpha_4$.

For the van der Waals potential, the first order correction terms are given by $c_{kl} \approx [b^l_k + b^k_l] - [a_{lk} + a_{kl}]$, and the condition $C(x)_{ij} = c(x)$ requires all $a_{ij}$ and $b^i_j$ to be equal.
The interpretation of the attractive interaction terms is as in the previous paragraph.
The volume exclusion terms $b^i_j$ lead to a decrease in $\alpha_4$, just as repulsive interactions (this is to be expected as volume exclusion is a kind of repulsive interaction).

Notice that from the formula (\ref{eq:alpha_4}), one might conjecture that the sensitivity $\alpha_4$ vanishes for $c(x) = \beta^{-1}$ and that it becomes negative for $c(x) >  \beta^{-1}$, irrespective of the value of $\alpha_4^{\mathrm{id}}$.
However, in such cases the hypothesis of $c(x)$ being small breaks down, and the analysis is no longer valid.
Going beyond the first order approximation is required.
However, even for the given example CRN with only 5 chemicals and under the hypothesis that $C(x)_{ij} = c(x)$, the symbolic computation of $A(x)$ based on the formula $A(x) = g_Y(y)U^* [U g_Y(y) U^*]^{-1} U$ or based on the expression $A(x)^i_j = \left\langle \dd x^i , \pi\left( \frac{\partial}{\partial x^j} \right) \right\rangle$ becomes rather intractable.
To explore the general situation, we therefore resort to numerical experiments.

\subsection{Numerical results}

The effects of a deformation of the metric $g^{\mathrm{id}}_{X}$ by $C$ as $g_X = g^{\mathrm{id}}_{X} + C$ beyond the first order in $C$ can be analyzed numerically based on the expression (\ref{eq:A_full}), $A(x) = g_Y(y)U^* [U g_Y(y) U^*]^{-1} U$, for the matrix of absolute sensitivities.
In order to get an intuition for the effects of the deformation, the state space $X$ is sampled randomly\footnote{
Each of the sampled points is assumed to be an equilibrium point.
This is possible due to the foliation of $X$ by equilibrium manifolds $\V$.

} in the range $\log x^i \in [-6,2]$, drawn from a uniform distribution.
The tensor components $C_{ij}$ with $i \geq j$ are drawn from a uniform distribution $C_{ij} \in [-c_{\textrm{max}},c_{\textrm{max}}]$ with variable $c_{\textrm{max}}$, and $C_{ij} = C_{ji}$ for $i <j$.
After numerically confirming that $g_X(x)$ is positive definite, the resulting $\alpha_4$ is recorded, otherwise it is discarded.
The distributions for $\alpha_4$ obtained from 100,000 data points each are shown in Figure \ref{fig:numerics}, with more Figures to be found in the Supplementary Material, Section \ref{SMsec:Figs}.
The respective mean values, standard deviations, fractions of negative values, and fractions of values $>1$ are recorded in Table \ref{table:numerics}.
\begin{figure}[ht]
    \centering
    \includegraphics[scale=0.4]{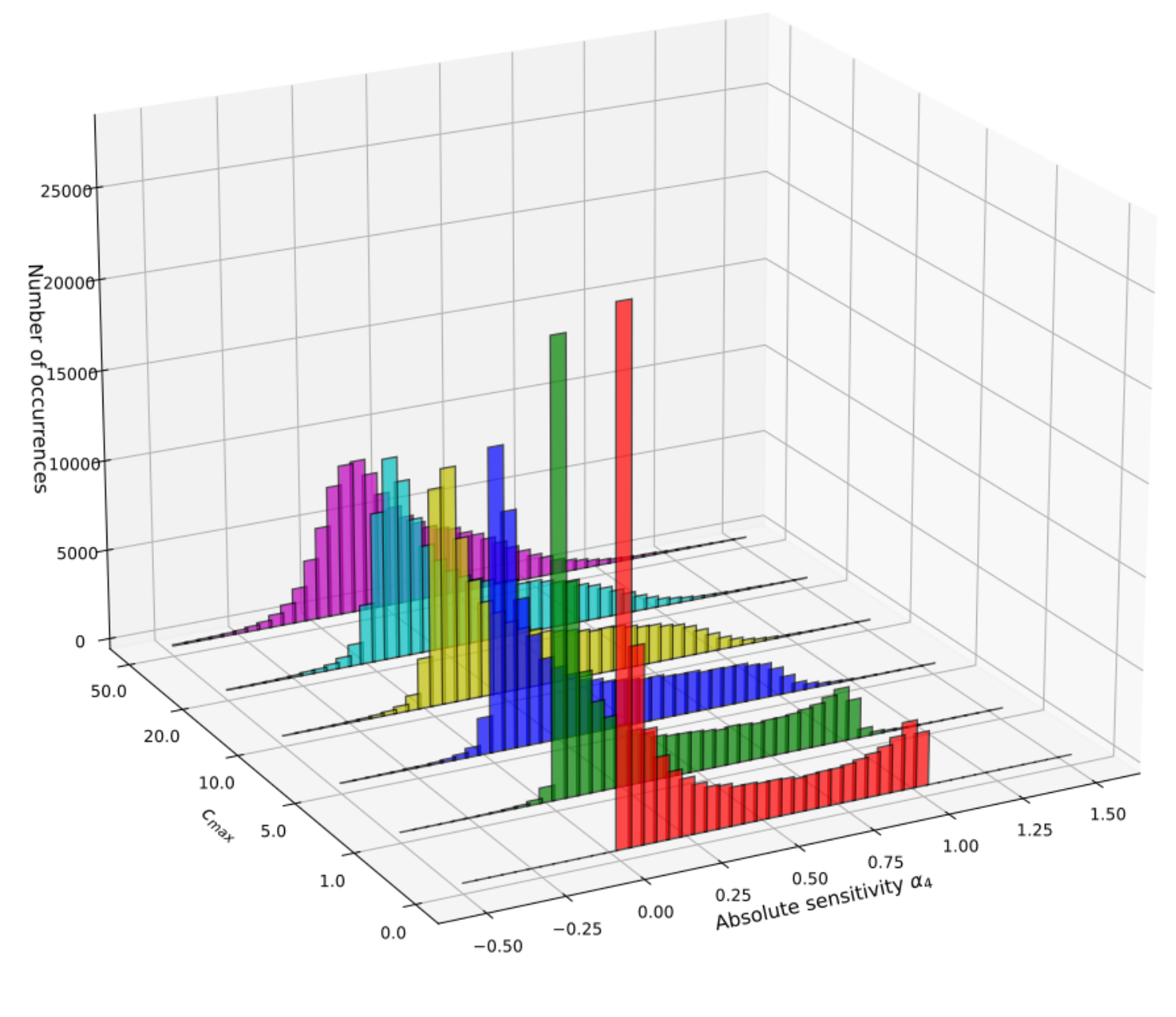}
    \caption{Distributions of the absolute sensitivity $\alpha_4$ for randomly sampled concentration vectors in $\log x^i \in [-6,2]$, and random deformation tensor components $C_{ij} \in [-c_{\textrm{max}},c_{\textrm{max}}]$ for $i \geq j$ and $C_{ij} = C_{ji}$ for $i <j$.
    Each histogram consists of 100,000 data points.
    }
    \label{fig:numerics}
\end{figure}
\begin{table}[htbp]
\footnotesize
  \caption{Mean values $\mu (\alpha_4)$, standard deviations $\sigma (\alpha_4)$, fraction of negative values, and fraction of values $>1$ for the distributions shown in Figure \ref{fig:numerics}.
    }\label{table:numerics}
    \begin{center}
    \begin{tabular}{|c||c|c|c|c|}
        \hline
        $c_{\textrm{max}}$ & $\mu (\alpha_4)$ & $\sigma (\alpha_4)$ & Fraction of $\alpha_4 < 0$ & Fraction of $\alpha_4 > 1$\\
        \hline
        \hline
        0.0 & 0.314 & 0.334 & 0.000 & 0.000 \\
        \hline
        1.0 & 0.309 & 0.331 & 0.034 & 0.012 \\
        \hline
        5.0 & 0.301 & 0.313 & 0.056 & 0.015 \\
        \hline
        10.0 & 0.295 & 0.301 & 0.066 & 0.015 \\
        \hline
        20.0 & 0.288 & 0.289 & 0.079 & 0.014 \\
        \hline
        50.0 & 0.276 & 0.274 & 0.099 & 0.012 \\
        \hline
    \end{tabular}
    \end{center}
\end{table}

It is interesting to note that the distribution for $c_{\textrm{max}} = 0$, i.e., for the ideal geometry resulting from $\phi^{\mathrm{id}}(x)$, is skewed towards $\alpha_4 = 0$, with a mean value of $\mu(\alpha_4)=0.314$ and a standard deviation of $\sigma(\alpha_4) = 0.334$.
With increasing interaction strength, both the mean values and the standard deviations are slightly decreasing.
\paragraph{Negative sensitivity}
The fraction of negative values for $\alpha_4$ is zero for $c_{\textrm{max}} = 0$, which follows from the general theory.
With increasing interaction strength this fraction is increasing, with a value of $0.099$ for $c_{\textrm{max}} = 50.0$.
This is highly unusual behavior as it corresponds to negative self-regulation, and yet it seems to occur with a rather high chance under non-ideal thermodynamical conditions.
\paragraph{Hypersensitivity} 
The other scenario which is not found under ideal conditions is $\alpha_4 > 1$. 
It corresponds to hypersensitivity.
The fraction of such behavior is considerably less than for negative sensitivity for the IDH enzyme for all examined values of $c_{\textrm{max}}$.
This underlines the tendency of the CRN for the IDH enzyme to be rather robust in its concentration than sensitive.
Moreover, the fraction is the largest for $c_{\textrm{max}} = 5.0$ and $c_{\textrm{max}} = 10.0$ with a value of $0.015$ and decreases for higher values of $c_{\textrm{max}}$.
\paragraph{Extreme behavior}
Figure \ref{fig:numerics} does not show the tails of the distributions but with increasing interaction strength, extremely large magnitudes of the sensitivity values $\alpha_4$ become increasingly likely.
Without analyzing these effects in this theoretical article, in the Supplementary Material, Section \ref{SMsec:extreme}, we give examples of concentration vectors $x$ and the interaction tensors $C$ which produce values $\alpha_4 <-4$ and $\alpha_4 > 4$
Such values correspond to very strong negative self-regulation and clearly pronounced hypersensitivity, respectively.

\subsection{Upshot}
The important result of this section is that the ideal behavior $\ai \in [0,1]$ can be violated as a result of thermodynamical interactions.
The numerical example provides strong evidence that this is a common occurrence rather than exotic behavior.
This suggests the possibility that biologically relevant effects such as negative feedback ($\ai < 0$) and hypersensitivity ($\ai >1$) might have a thermodynamical origin.
We are looking forward to conduct more thorough numerical studies in the future.

\section{Discussion} \label{sec:discussion}

In this article, the Riemannian geometric aspects of the information geometry of CRN have been developed.
The central result is that the commonly used parametrizations of the respective equilibrium spaces and the Legendre transformation between spaces on the concentration side and the chemical potential side are isometries with respect to the Riemannian metrics given by the Hessians of convex potential functions.
Moreover, it was made precise that the metrics encode the infinitesimal information for the Legendre transformations and are therefore equivalent to the respective derivatives.
Using this setup allowed to prove an informational geometrical Cramer-Rao bound for CRNs by comparing two Riemannian metric tensors.

In the second part of the paper, the geometrical setup was applied to get a better understanding of absolute sensitivity, generalizing the theory from \cite{loutchko2024cramer}.
There are two central results:
First, a linear algebraic formula for the matrix of absolute sensitivities was proven.
This is useful for applications such as numerical experiments and data analysis.
Second, absolute sensitivity was defined in a global and coordinate-free manner as a bundle morphism.
By using the Riemannian geometric tools, it was proven that absolute sensitivity is the $g_X$-orthogonal projection operator from the tangent bundle $TX$ of the concentration space to the tangent bundle $T\V$ of the equilibrium manifold.
From this geometric characterization, a closed formula for the first order corrections to the ideal, i.e., quasi-thermostatic case, was derived, and a vanishing criterion for the correction terms was formulated.

The results were illustrated by using the reaction network for the core module of the IDHKP-IDH glyoxylate bypass regulation system as an example.
This system is known to exhibit concentration robustness for the IDH enzyme, which becomes absolute under certain irreversibility conditions.
Thus the focus was put on the absolute sensitivity of the IDH enzyme.
For the quasi-thermostatic case it has been computed and analyzed in \cite{loutchko2024cramer}, and here, the first order correction was derived explicitly so that the effects of attractive and repulsive interactions could be discussed.
Beyond the first order approximation, numerical experiments were performed, for random values of the deformation tensor and for random concentration values in a physiological range.
It was found that interactions decrease the mean values and the standard deviations of the absolute sensitivity for the IDH enzyme on average.
Moreover, with increasing interaction strength, the occurrence of negative self-feedback of the IDH enzyme was found to be increasing.
Cases of hypersensitivity were detected but they were less prevalent than negative self-feedback by a factor of about 10.

These results show how non-ideal thermodynamical behavior can lead to effects which usually require highly nonlinear kinetics in the ideal case to occur.
As cellular environments are highly crowded and interactive, the understanding of non-ideal behavior is very important in the modeling of biochemical CRNs.
In this regard, this article demonstrates how such effects can be analyzed by using the information geometrical framework, and how geometrical insight can lead to a clear understanding of the effects captured by deformations of the metric.
This in turn can lead to efficient numerical algorithms, or even vanishing criteria without the need for computation.
It is worth stressing that keeping the global geometrical viewpoint as long as possible is crucially important to this approach.
This is because brute-force symbolic computations in CRN theory can get out of hand very quickly, whereas numerical studies do not provide enough rigorous insights.

Studying the deformations beyond the first order is rather difficult and will require more technical tools from geometry.
Moreover, the compatibility of non-ideal equilibrium states with kinetics is an important open question for future research.
In this regard, it is worth noting that the large deviation rate functions which control the geometry of the steady states in nonequilibirum CRNs have been found \cite{anderson2015lyapunov} to have the same local structure as the deformed geometry studied in the present article.
This opens up the possibility to extend the theory developed here to the nonequilibrium situation which will be explored in future work.

\section*{Acknowledgments}

This research is supported by JST (JPMJCR2011, JPMJCR1927) and JSPS (19H05799).
Y. S. receives financial support from the Pub-lic\verb|\|Private R\&D Investment Strategic Expansion PrograM (PRISM) and programs for Bridging the gap between R\&D and the IDeal society (society 5.0) and Generating Economic and social value (BRIDGE) from Cabinet Office.
We thank Praful Gagrani and Atsushi Kamimura for fruitful discussions.

\bibliographystyle{siamplain}
\bibliography{references}

\unappendix
\pagebreak
\begin{center}
\textbf{\large Supplementary Materials: Information geometry of chemical reaction networks: Cramer-Rao bound and absolute sensitivity revisited}
\end{center}
\setcounter{equation}{0}
\setcounter{figure}{0}
\setcounter{table}{0}
\setcounter{section}{0}
\setcounter{page}{1}
\makeatletter
\renewcommand{\theequation}{S\arabic{equation}}
\renewcommand{\thefigure}{S\arabic{figure}}
\renewcommand{\thesection}{S\arabic{section}}

\section*{Contents}

Section \ref{SMsec:X} elaborates on the Example \ref{ex:quasiTS} by proving that the van der Waals potential $\phi$ is strictly convex on the given domain $X$.

Section \ref{SMsec:example} is an addendum to the Section \ref{sec:IDH} in the main text, which deals with the core module of the IDHKP-IDH glyoxylate bypass regulation system.

\section{The convex manifold $X$ for the van der Waals gas} \label{SMsec:X}

In Example \ref{ex:quasiTS} in the main text, it is claimed that the domain 
\begin{equation*}
        X = \left\{ x \in \R^n_{>0}: x^k < \frac{1}{nC_k} \right\},
\end{equation*}  
with $B := \max_{k,l} \left\{ b^k_l \right\}$ and $C_k:= \max_{l} \left\{ a_{kl} + a_{lk} + 2B, 4B \right\}$, equipped with the function
\begin{equation*}
        \phi(x) = \phi^{\mathrm{id}}(x) - \sum_{i=1}^n x^i \log \left[ 1 - \sum_{j=1}^n b^i_{j} x^j \right] - \sum_{i,j =1}^n a_{ij}x^i x^j,
\end{equation*}
satisfies the conditions of a concentration space from Section \ref{sec:Hessian_setup}.
This requires to show that
\begin{enumerate}
\label{cond:X}
    \item $X$ is a convex open subset of $\R^n_{>0}$.
\label{cond:def}
    \item The function $\phi(x)$ is defined on $X$.
\label{cond:phi}
    \item The function $\phi(x)$ is strictly convex on $X$.
\end{enumerate}
The first condition is obvious.
The second condition requires that $1 - t^i > 0$ for all $i$, where
\begin{equation*}
    t^i:= \sum_{j=1}^n b^i_{j} x^j.
\end{equation*}
Let $t : = \sum_{j=1}^n x^j $.
Then 
\begin{equation*}
    t < \frac{1}{n} \sum_{k=1}^n \frac{1}{C_k} < \frac{1}{n} \sum_{k=1}^n \frac{1}{2B} = \frac{1}{2B}
\end{equation*}
and thus $Bt < \frac{1}{2}$.
But
\begin{equation*}
    t^i < B \sum_{j=1}^n x^j = Bt
\end{equation*}
and therefore $\frac{1}{2} < 1 - t^i < 1$.

The third condition is verified by showing that the Hessian matrix of $\phi$ is positive semidefinite.
As it is symmetric by definition, it remains to show that all its eigenvalues are positive, which is achieved by showing that it is strictly diagonally dominant with positive diagonal entries.
The Hessian of $\phi$ is given by
\begin{equation*}
    {g_X}_{kl} = \frac{\partial^2 \phi(x)}{\partial x_l \partial x_k} = {g_X^{\mathrm{id}}}_{kl} + c_{kl}
\end{equation*}
with $g^{\mathrm{id}}_X = \XX$ and 
\begin{equation*}
    c_{kl} = - \sum_{i=1}^n \frac{x^i b^i_k b^i_l}{\left(1 - t^i \right)^2} + \frac{b^l_k}{1 - t^l} + \frac{b^k_l}{1 - t^k }  - [a_{lk} + a_{kl}].
\end{equation*}
The following claim yield the desired diagonal dominance
\begin{claim*}
For all $k,l$, the inequality
\begin{equation*}
    \frac{1}{n} \frac{1}{x^k} > |c_{kl}|
\end{equation*}
holds.
\end{claim*}
\begin{proof}
    If $c_{kl} > 0$, then
\begin{equation*}
|c_{kl}| < \frac{b^l_k}{1 - t^l} + \frac{b^k_l}{1 - t^k } \leq \frac{2B}{1 - Bt} < 4B \leq C_k < \frac{1}{n} \frac{1}{x^k},
\end{equation*}
where the last inequality follows from the condition on the domain $X$.

If $c_{kl} < 0$, then
\begin{align*}
|c_{kl}| &< [a_{lk} + a_{kl}] + \sum_{i=1}^n \frac{x^i b^i_k b^i_l}{\left(1 - t^i \right)^2} \leq [a_{lk} + a_{kl}] + B \frac{Bt}{(1-Bt)^2}  \\
&< [a_{lk} + a_{kl}] + 2B \leq  C_k < \frac{1}{n} \frac{1}{x^k}.
\end{align*}
This proves the claim.
\end{proof}
Now, if $c_{kk} < 0$, then, by summing over $l=1,\dotsc,n$, the claim gives
\begin{equation*}
    \frac{1}{x^k} > \sum_{l} |c_{kl}| =  \sum_{l \neq k} |c_{kl}| - c_{kk},
\end{equation*}
which yields
\begin{equation*}
    \frac{1}{x^k} + c_{kk} > \sum_{l \neq k} |c_{kl}|.
\end{equation*}
If $c_{kk} \geq 0$, then the claim leads to
\begin{equation*}
    \frac{1}{x^k} + c_{kk} \geq  \frac{1}{x^k} > \sum_{l \neq k} |c_{kl}|.
\end{equation*}
This shows the validity of condition (3).

\section{The core module of the IDHKP-IDH glyoxylate bypass regulation system} \label{SMsec:example}

The CRN of the core module of the IDHKP-IDH (IDH = isocitrate dehydrogenase, KP = kinase-phosphatase) glyoxylate bypass regulation system is shown in Scheme (\ref{SMeq:ACR}).
Hereby, I is the IDH enzyme, $\textrm{I}_{\textrm{p}}$ is its phosphorylated form, and E is the bifunctional enzyme IDH kinase-phosphatase.
It is known to obey approximate concentration robustness in the IDH enzyme \cite{laporte1985compensatory}.
\begin{equation} \label{SMeq:ACR}
\begin{tikzcd}
    \textrm{E} + \textrm{I}_{\textrm{p}}  \ar[r, rightharpoonup, shift left=.35ex,"k_1^+"] & \textrm{EI}_{\textrm{p}} \ar[l,shift left=.35ex, rightharpoonup,"k_1^-"] \ar[r, rightharpoonup, shift left=.35ex,"k_2^+"] & \textrm{E} + \textrm{I} \ar[l,shift left=.35ex, rightharpoonup,"k_2^-"] \\
    \textrm{EI}_{\textrm{p}} + \textrm{I} \ar[r, rightharpoonup, shift left=.35ex,"k_3^+"] & \textrm{EI}_{\textrm{p}}\textrm{I} \ar[l,shift left=.35ex, rightharpoonup,"k_3^-"] \ar[r, rightharpoonup, shift left=.35ex,"k_4^+"] & \textrm{EI}_{\textrm{p}} + \textrm{I}_{\textrm{p}} \ar[l,shift left=.35ex, rightharpoonup,"k_4^-"]
\end{tikzcd}
\end{equation}
In \cite{shinar2009robustness,shinar2010structural} it was shown that if the CRN obeys mass action kinetics and the rate constants $k_2^-$ and $k_4^-$ are zero then the concentration robustness in the IDH enzyme I holds exactly.

For notational convenience, the abbreviations $X_1 = \textrm{E}, X_2 = \textrm{I}_{\textrm{p}}, X_3 = \textrm{EI}_{\textrm{p}}, X_4 = \textrm{I}, X_5 = \textrm{EI}_{\textrm{p}}\textrm{I}$ are introduced which gives the 
\begin{equation} \label{SM:CRN2}
\begin{tikzcd}
    X_1 + X_2  \ar[r, rightharpoonup, shift left=.35ex,"k_1^+"] & X_3 \ar[l,shift left=.35ex, rightharpoonup,"k_1^-"] \ar[r, rightharpoonup, shift left=.35ex,"k_2^+"] & X_1 + X_4 \ar[l,shift left=.35ex, rightharpoonup,"k_2^-"] \\
    X_3 + X_4 \ar[r, rightharpoonup, shift left=.35ex,"k_3^+"] & X_5 \ar[l,shift left=.35ex, rightharpoonup,"k_3^-"] \ar[r, rightharpoonup, shift left=.35ex,"k_4^+"] & X_3 + X_2. \ar[l,shift left=.35ex, rightharpoonup,"k_4^-"]
\end{tikzcd}
\end{equation}
The variables $x_i, i =1,\dotsc,5$ are used for the respective concentrations of the chemicals $X_i$.

\subsection{Computation of $\alpha_4$} \label{app:details}

The focus lies on the computation of the absolute sensitivity $\alpha_4$ as it is the sensitivity of the chemical exhibiting (approximate) concentration robustness.
In \cite{loutchko2024cramer}, Section 5, the absolute sensitivity of I was calculated for the quasi-thermostatic case, i.e., for $\V^{\mathrm{id}}$ resulting from the convex function $\phi^{\mathrm{id}}$, with nonvanishing $k_2^-$ and $k_4^-$.

It is assumed that $c_{ij} = c$ for all $i,j$ and thus the Hessian metric $g_X$ becomes
\begin{equation*}
    g_X =  {g_X^{\mathrm{id}}} + C = {g_X^{\mathrm{id}}} + c \begin{bmatrix}
        1 & 1 & 1 & 1 & 1 \\
        1 & 1 & 1 & 1 & 1\\
        1 & 1 & 1 & 1 & 1\\
        1 & 1 & 1 & 1 & 1\\
        1 & 1 & 1 & 1 & 1 
    \end{bmatrix}.
\end{equation*}
The first order correction in $c$ to $\alpha_4^{\mathrm{id}}$ is now calculated.

For the CRN in Scheme \ref{SM:CRN2}, the absolute sensitivity $\alpha_4$ is computed by using Theorem \ref{thm:abs_sens_final}.
As the first step, the tangent space $T_x \V$ is determined.
The stoichiometric matrix of the CRN is given by
\begin{equation*}
S=  \begin{bmatrix}
        -1 & 1 & 0 & 0 \\
        -1 & 0 & 0 & 1 \\
        1 & -1 & -1 & 1 \\
        0 & 1 & -1 & 0 \\
        0 & 0 & 1 & -1 
    \end{bmatrix}
\end{equation*}
and the kernel of $S^*$ is spanned by $v_1 := (-1,1,0,1,1)$ and $v_2 := (1,0,1,0,1)$.
According to Lemma \ref{lem:comp}, the tangent space at $x$ is given by $T_x \V \simeq g_Y \Ker[S^*]$.
To first order in $c$, the metric $g_Y$ is
\begin{equation*}
    g_Y = \left[{g_X^{\mathrm{id}}} + C \right]^{-1} =  g_Y^{\mathrm{id}} - g_Y^{\mathrm{id}} C g_Y^{\mathrm{id}},
\end{equation*}
where $g_Y^{\mathrm{id}}$ is represented by the diagonal matrix $\X$.
Then $T_x \V$ is spanned by $w_1 := g_Y v_1$ and $w_2 := g_Y v_2$.
The linear combination
\begin{equation*}
    w_2' := w_2 - \frac{\langle w_2, w_1 \rangle_{g_X}}{\langle w_1, w_1 \rangle_{g_X}} w_1
\end{equation*}
is by construction orthogonal to $w_1$ and the projection of any vector $v$ to $T_x \V$ is given by
\begin{equation*}
    \pi(v) = \frac{\langle v, w_1 \rangle_{g_X}}{\langle w_1, w_1 \rangle_{g_X}} w_1 +  \frac{\langle v, w_2 \rangle_{g_X}}{\langle w_2, w_2 \rangle_{g_X}} w_2
\end{equation*}
For $\alpha_4 = \left[\pi\left(\frac{\partial }{\partial x^4}\right)\right]_4$, a calculation in SageMath \cite{sagemath}, cf. Section \ref{sec:SMcalc} over the ring
\begin{equation*}
    \faktor{\mathbb{R}(x^1,\dotsc,x^5)[c]}{\left(c^2\right)}
\end{equation*}
yields
\begin{equation*}
    \alpha_4 = \alpha_4^{\mathrm{id}} (1 - c\beta),
\end{equation*}
where 
\begin{equation*}
    \alpha_4^{\mathrm{id}} = \frac{1}{1+r},
\end{equation*}
is the absolute sensitivity for the quasi-thermostatic case, i.e., for $\V^{\mathrm{id}}$ resulting from the convex function $\phi^{\mathrm{id}}$ and $r$ is the ratio
\begin{equation*}
    r = \frac{\left(x_2 + x_5 \right) \left(x_1 + x_3 \right) + x_1 \left( x_3 + 3x_5\right) +x_2x_5}{x_4 \left(x_1 + x_3 + x_5 \right)}.
\end{equation*}
The the correction term $\beta$ is given by
\begin{equation*}
    \beta = \frac{(x_2 + x_4) (x_1x_3 + 4 x_1x_5 + x_3 x_5) }{(x_1+x_3+x_5) (x_1 + x_2 + x_4 + x_5)- (x_5-x_1)^2}.
\end{equation*}
Notice that the denominator $(x_1+x_3+x_5) (x_1 + x_2 + x_4 + x_5)- (x_5-x_1)^2 = x_1x_2 + x_1x_3 + x_2x_3 + x_1x_4 + x_3x_4 + 4x_1x_5 + x_2x_5 + x_3x_5 + x_4x_5$ is always positive and hence
\begin{equation*}
    \beta > 0.
\end{equation*}

\newpage

\subsection{Additional figures} \label{SMsec:Figs}

The histograms from Figure \ref{fig:numerics} in the main text are shown here from different perspectives.
\begin{figure}[ht]
    \centering
    \includegraphics[scale=0.31]{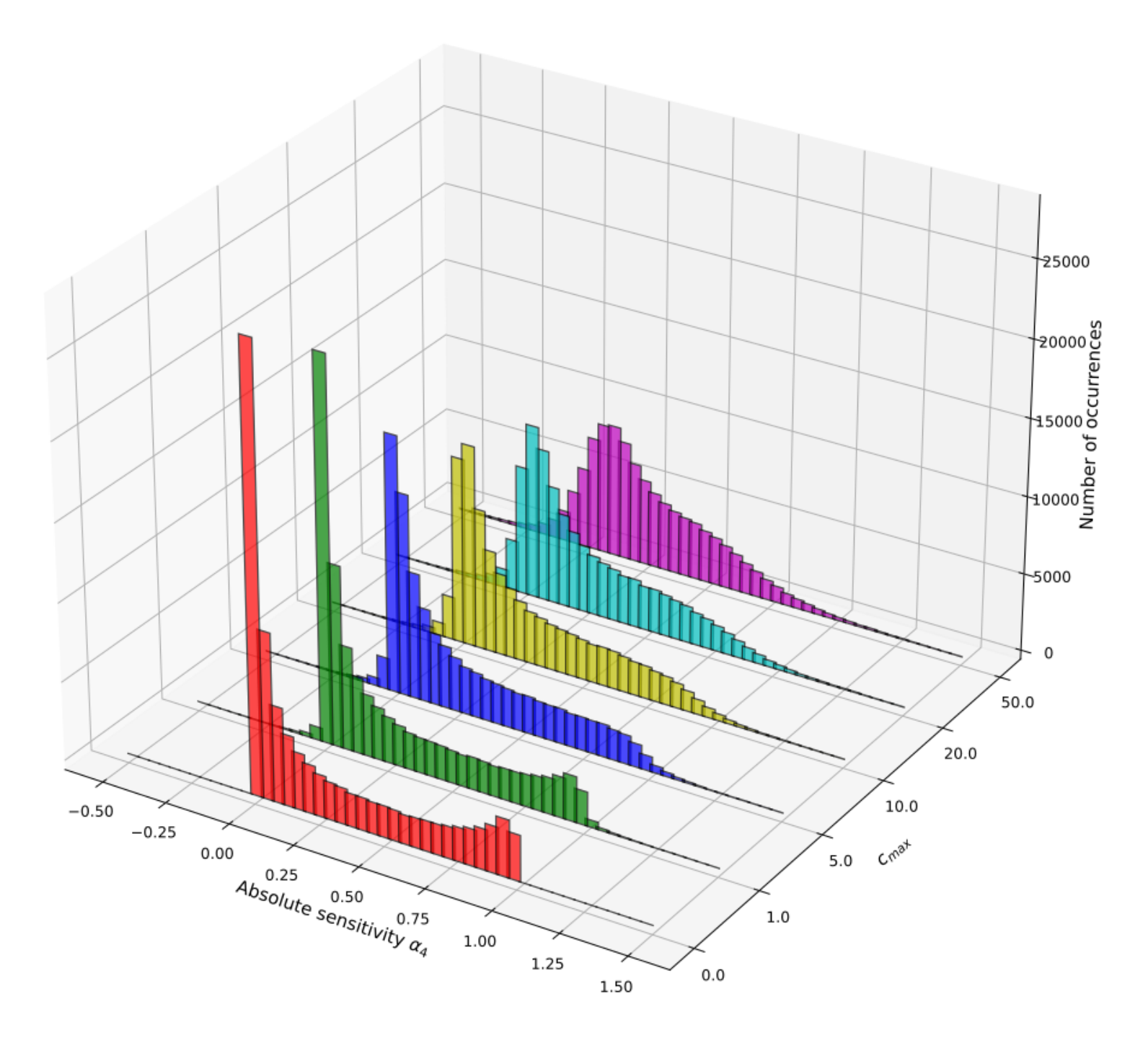}
    \caption{The histograms from Figure \ref{fig:numerics} from a different perspective.
    }
    \label{fig:numerics1}
\end{figure}
\begin{figure}[ht]
    \centering
    \includegraphics[scale=0.31]{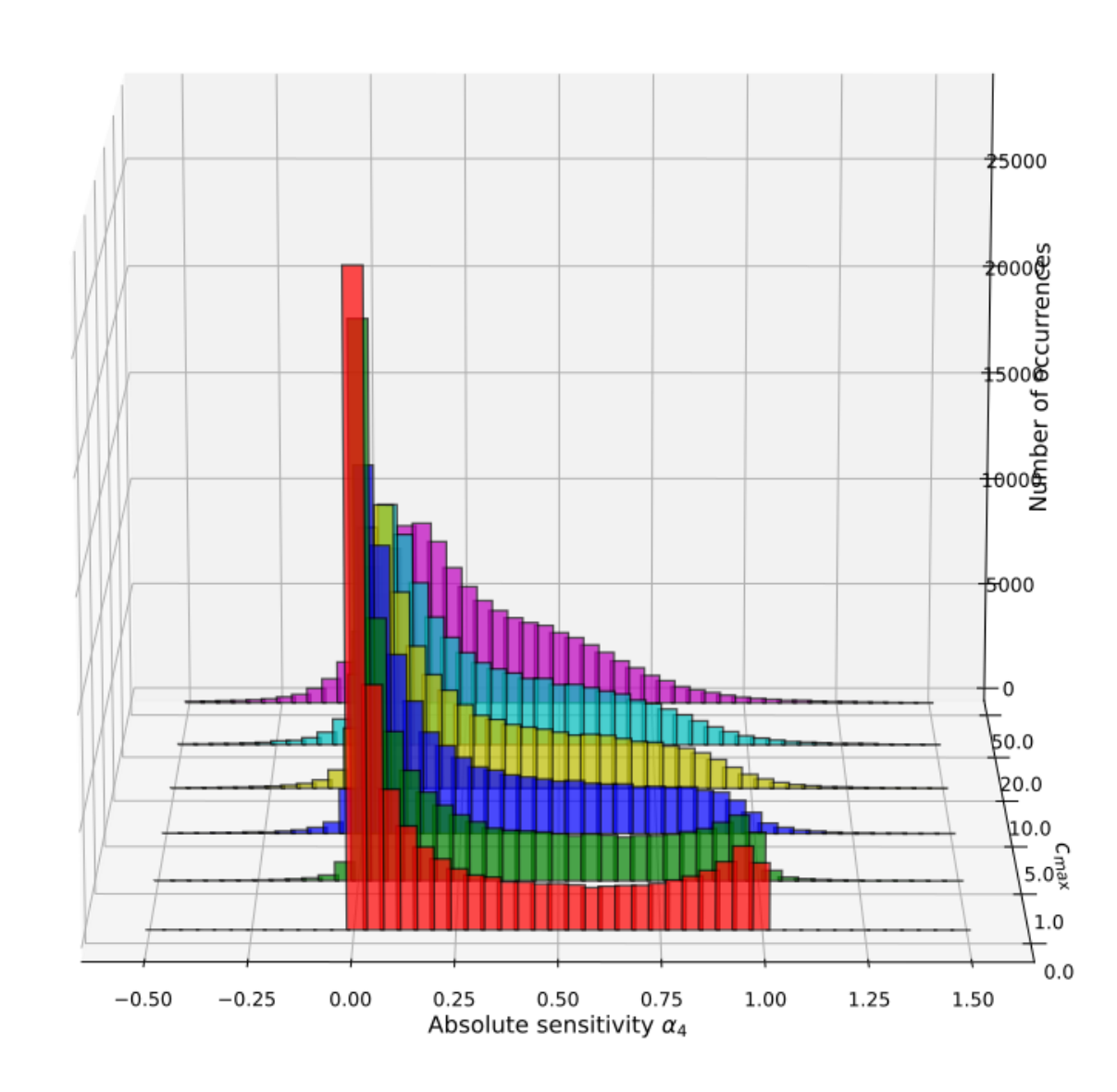}
    \caption{The histograms from Figure \ref{fig:numerics} from a different perspective.
    }
    \label{fig:numerics2}
\end{figure}

\newpage

\subsection{Extreme examples} \label{SMsec:extreme}

\paragraph{Hypersensitivity}
The following data leads to an absolute sensitivity value of $\alpha_4 \approx 4.05$:
\begin{equation*}
    x = (0.0330922,  6.89637913, 0.0206041,  2.94561781, 0.00851101)
\end{equation*}
and $C =$
\begin{equation*}
    \begin{pmatrix}
7.8317431 & -1.2979878 &  6.749366  & -4.08968516 & 16.77522956 \\
-1.2979878 & 19.37720745 &  7.56087486 &  12.93634691 & 14.89114232 \\
6.749366  &  7.56087486 & 10.06408641 & 8.62541839 & 15.93640615 \\
-4.08968516 & 12.93634691 & 8.62541839 & 10.22790791 & 19.4292935 \\
16.77522956  & 14.89114232 & 15.93640615 & 19.4292935 &  -5.86803736
\end{pmatrix}.
\end{equation*}

\paragraph{Negative self-feedback}
The following data results in a value of $\alpha_4 \approx -5.04$ for the absolute sensitivity:
\begin{equation*}
    x = (0.00361443, 0.05000115, 0.00515046, 0.19451776, 0.00579231)
\end{equation*}
and $C=$
\begin{equation*}
    \begin{pmatrix}
-2.32733249  & 16.92656819 &  -4.02679579 & -16.57677126 &  13.98898945\\
16.92656819  & -9.68791321 & -12.49816251 &  6.88887821 &  2.31959764\\
-4.02679579 & -12.49816251 & -14.96354673 &  12.10127076 & -11.29262967 \\
-16.57677126  &  6.88887821 & 12.10127076 &  5.70349432 & -0.75812926\\
13.98898945 &  2.31959764 & -11.29262967 &  -0.75812926 &  9.39908501
\end{pmatrix}.
\end{equation*}

\subsection{SageMath code} \label{sec:SMcalc}

The following code was used to compute $\alpha^{\mathrm{id}}$ and $\beta$ in SageMath \cite{sagemath}:

\begin{lstlisting}
In[1]:    
#work to first order in c, i.e., modulo c^2.

R.<x1,x2,x3,x4,x5> = PolynomialRing(QQ,5) ; R
F = R.fraction_field() ; F
G.<c> = F[]
Q = QuotientRing(G, G.ideal(c^2))
#R.<x1,x2,x3,x4,x5,c> = PolynomialRing(QQ,6) ; R
#Q = R.fraction_field() ; Q

#making the metrics
A = Matrix(Q,[[x1^(-1),0,0,0,0],[0,x2^(-1),0,0,0],[0,0,x3^(-1),0,0],[0,0,0,x4^(-1),0],[0,0,0,0,x5^(-1)]])
C = Matrix(Q,[[c,c,c,c,c],[c,c,c,c,c],[c,c,c,c,c],[c,c,c,c,c],[c,c,c,c,c]])
gX = A + C
AA = Matrix(Q,[[x1,0,0,0,0],[0,x2,0,0,0],[0,0,x3,0,0],[0,0,0,x4,0],[0,0,0,0,x5]])
gY = AA - AA * C * AA

#calculating the sensitivity
v1 = vector(Q,[-1,1,0,1,1])
v2 = vector(Q,[1,0,1,0,1])
e4 = vector(Q,[0,0,0,1,0])
#now making the orthogonal basis for T_xV = span(gY*v1,gY*v2)
vv1 = gY*v1
vv2 = gY*v2
w1 = vv1
w2 = vv2 - (vv2*gX*w1)/(w1*gX*w1)*w1
#calculating pi(e4)
pi = (e4*gX*w1)/(w1*gX*w1)*w1 + (e4*gX*w2)/(w2*gX*w2)*w2

#calculating $\alpha = \alpha^{id}(1 - c*\beta)$
alpha4 = (e4*pi)
f = alpha4.lift()
alphaid = f(c=0)
#g = -\alpha^{id}\beta
g = f.derivative(c)
#print(g)
beta = -g/alphaid
# -beta from the text
betaText = (x2 + x4)*(x1*x3 + 4*x1*x5 + x3*x5)/((x1 + x3 + x5)*(x1 + x2 + x4 + x5) - (x5 - x1)^2)
print(beta == betaText)
print(beta)

Out[1]:
True
(x1*x2*x3 + x1*x3*x4 + 4*x1*x2*x5 + x2*x3*x5 + 4*x1*x4*x5 + x3*x4*x5)/(x1*x2 + x1*x3 + x2*x3 + x1*x4 + x3*x4 + 4*x1*x5 + x2*x5 + x3*x5 + x4*x5)
\end{lstlisting}

\end{document}